\documentclass[11pt,letterpaper]{article}
\usepackage[margin=1.0in]{geometry}
\usepackage{ifpdf}
\ifpdf
    \usepackage[pdftex]{graphicx}
    \graphicspath{{figs/}{matlab/}}   
    \usepackage[update]{epstopdf}
\else
	\usepackage{graphicx}
	\graphicspath{{figs/}{matlab/}}   
\fi
\usepackage[title]{appendix}
\usepackage{wrapfig,bbm}
\usepackage{enumitem,color}
\usepackage{amsthm}
\newtheorem{theorem}{Theorem}
\newtheorem{lemma}{Lemma}
\newtheorem{corollary}{Corollary}

\newtheorem{prop}{Proposition}

\usepackage{amsmath}
\usepackage{amssymb}
\usepackage{cite,setspace}

\newcommand{\Ex}{ \mathbb{E}  }

\begin{document}
\title{\vspace{-0.5cm}\LARGE Matched Multiuser Gaussian Source Channel Communications \\via Uncoded Schemes}
\author{Chao Tian, Jun Chen, Suhas N. Diggavi and Shlomo Shamai
\thanks{This work is presented in part at the 2015 IEEE International Symposium on Information Theory, Hong Kong, China, June 2015. The work of S. Shamai has been supported by the Israel Science Foundation (ISF), and by the S. and N. Grand Research Fund. The work of S N. Diggavi was supported in part by NSF grants CCF-1314937  and CCF-1514531}}
\maketitle

\abovecaptionskip=0.2cm
\abovedisplayskip=3pt
\belowdisplayskip=4pt
\begin{abstract}
We investigate whether uncoded schemes are optimal for Gaussian sources on multiuser Gaussian channels. Particularly, we consider two problems:  the first is to send correlated Gaussian sources on a Gaussian broadcast channel where each receiver is interested in reconstructing only one source component (or one specific linear function of the sources) under the mean squared error distortion measure; the second is to send correlated Gaussian sources on a Gaussian multiple-access channel, where each transmitter observes a noisy combination of the source, and the receiver wishes to reconstruct the individual source components (or individual linear functions) under the mean squared error distortion measure. It is shown that when the channel parameters match certain general conditions, the induced distortion tuples are on the boundary of the achievable distortion region, and thus optimal. Instead of following the conventional approach of attempting to characterize the achievable distortion region, we ask the question whether and how a match can be effectively determined. This decision problem formulation helps to circumvent the difficult optimization problem often embedded in region characterization problems, and it also leads us to focus on the critical conditions in the outer bounds that make the inequalities become equalities, which effectively decouple the overall problem into several simpler sub-problems. Optimality results previously unknown in the literature are obtained using this novel approach. As a byproduct of the investigation, novel outer bounds are derived for these two problems. 
\end{abstract}

\section{Introduction}
\label{sec:intro}
Although the source channel separation architecture is asymptotically optimal in the point-to-point communication setting \cite{Shannon:48} as well as several classes of multiuser communication settings (see {\em e.g.,} \cite{Tian:14} and references therein), uncoded schemes have several particularly attractive properties. Firstly, they have very simple encoders and decoders; secondly, they belong to the so-called zero-delay codes, which can avoid the long delay required to approach the asymptotic performance in the separation-based schemes; lastly, they are in fact optimal in some settings where the separation-based schemes are not (see {\em e.g.},  \cite{Goblick65}).

It was shown in \cite{Gastpar:03} that uncoded schemes are optimal when certain matching conditions involving the source probability distribution, the channel transition probability distribution, the channel cost function and the distortion measure function are satisfied. Though the focus in \cite{Gastpar:03} was mainly on the point-to-point setting, recent results \cite{Gastpar:08,Lapidoth:10BC,Lapidoth:10MAC,Tian:11} suggest that the concept of matching indeed carries over  to the multiuser case. In fact, in multiuser settings, matching may occur naturally when the distortion measure, the channel cost function and source distribution are all fixed, and the channel parameters, which represent physically meaningful quantities, satisfy certain conditions. In this work, we consider such matching, particularly, when the sources and the channels are Gaussian, the channel constraints are on the expected average signal powers, the distortion measure is the mean squared error (MSE), and only the channel parameters, such as the channel  amplification factors and the additive noise powers, are allowed to vary.

In this context, of interest is whether for a fixed source and fixed coding parameters, the distortion vector such induced is on the boundary of the achievable distortion region and thus optimal. More specifically, we seek to answer the following questions:
\begin{itemize}
\item Is there a set of  (explicitly) computable conditions that can be used to certify a fixed uncoded scheme to be optimal for a given source and channel pair?
\item If so, is there a non-trivial set of channels that satisfy such conditions for a given source and uncoded scheme pair?
\end{itemize}

We shall refer to this kind of channels as ``matched channels''; a dual question is to ask for ``matched sources'', however in the context of the problems considered here, the dual question is notationally more involved, and thus we choose to investigate them from the perspective of ``matched channels''. One can also ask for ``matched distortion measures'', similarly as the approach taken in \cite{Gastpar:03}, however in the Gaussian setting, the MSE distortion is a practically more important and well-motivated case. The set of matched channels should be distinguished from the complete set of channels for which the given uncoded scheme is optimal. The former may be a strict subset of the latter, since these sufficient conditions for optimality in fact depend on the specific outer bounds that can be derived. Characterizing the latter region is naturally more difficult than answering the questions we posed above.

The two questions given above are in essence the two facets of the same question. Since we only provide conditions for matching, or in other words, sufficient conditions for a scheme to be optimal, the set of matched channels may in fact be empty. A trivial condition to answer the first question is simply an impossible one such that we would never be able to certify a channel to be matched. Thus the second question is important, and we show indeed for the two problems considered here, there are non-trivial channels that match the source and the uncoded scheme.

Traditionally, research in information theory asks for computable characterizations of a certain achievable region, for which we first derive an expression for an outer bound, and derive an expression for an inner bound, and then make comparison of them. This approach can be challenging because it usually involves optimization over a set of parameters, and solving such an optimization problem explicitly can be difficult. It is not clear whether the obstacle mainly stems from the intractable nature of the underlying communication problem, or it is mainly caused by the embedded optimization problem.

The aforementioned difficulty motivates the formulation of the first question, which is a decision problem instead of an optimization problem. An analogy of this situation can be found in computer science algorithm research, where instead of asking whether an optimization problem can be solved in polynomial time, an alternative question is asked whether a decision ({\em e.g.}, regarding a solution is above a threshold) can be made in polynomial time. Our problem formulation naturally leads to a different approach in the investigation. Instead of focusing on comparison of the inner bounds and outer bounds using their expressions, we focus on the necessary conditions that the outer bounds become tight, {\em i.e.,} the conditions when the information inequalities hold with strict equality. With fixed source and fixed coding parameters,  the coding vector can be substituted into the conditions, and the necessary and sufficient conditions for such equality can be derived. The outer bounds naturally provide certain ``decoupled'' conditions, which significantly simplify the overall task. Though this approach may have inherently been used by many researchers in the past, its effectiveness becomes particularly evident in our investigation of the joint source channel communication setting.

In the rest of the paper, we focus specifically on two joint source channel coding problems using the approach outlined above. The first problem is to send correlated Gaussian sources on a Gaussian broadcast channel where each receiver is interested in reconstructing only one source component
(or equivalently, one specific linear function of the source) under the MSE distortion measure. The second problem is to send correlated Gaussian sources on a Gaussian multiple-access channel, where each transmitter observes a noisy combination of the source, {\em i.e.,} a case of the vector CEO problem, and the receiver wishes to reconstruct the source components (or equivalently, linear functions of the source components) under the MSE distortion measure. General conditions for matching are derived, which provide new optimality results previously unknown in the literature. These results either include or generalize well-known existing results on the optimality of uncoded schemes in the multiuser setting. 
Particularly notable are the following cases:
\begin{itemize}
\item The first problem generalizes the two-user case considered in \cite{Lapidoth:10BC} and \cite{Tian:11} to the $M$-user case, for which we show that an uncoded scheme is optimal for a large set of sources and channels;  our results reveal that uncoded scheme can still be optimal when some source components are negatively correlated.
\item The results on the second problem includes as special cases the symmetric scalar Gaussian CEO problem \cite{Gastpar:08}, the problem of sending bivariate Gaussian sources on a Gaussian multiple-access channel \cite{Lapidoth:10MAC}, and sending remote (noisy) bivariate Gaussians on a Gaussian multiple-access channel \cite{LapidothWang:11}. Our results reveal that in addition to the symmetric case considered in \cite{Gastpar:08}, uncoded scheme is also optimal when the sensor observation quality is proportional to the channel quality.  These results also allow the sensor observations to have more general correlation structure and the observations to be noisy, thus extending the results in\cite{Lapidoth:10MAC} and \cite{LapidothWang:11}. When viewed from the perspective of computation, our result also provides new insights on the problem of computing linear functions of Gaussian random variables on the Gaussian multiple-access channels considered in \cite{NazerGastpar:07} and \cite{Vishwanath:12}.
\end{itemize}

Although we emphasize in this work the less conventional approach used to obtain the general matching conditions, during the process of this investigation, novel outer bounds are in fact derived for both problems beyond what are available in the literature. These new bounds rely on a technique motivated by a series of our previous works \cite{TianDiggaviShamai:09,Song:12,Khezeli:14(1),Khezeli:14(2)}, the origin of which can be further traced back to Ozarow \cite{Ozarow:80}. 

Notationally, we write for a source $S$ at time $n$ as $S[n]$, and a length-$N$ vector as $S^N$. For a set of quantities $(\alpha_1,\alpha_2,\ldots,\alpha_M)$, we write it in a (column) vector form as $\bar{\alpha}$ when its dimension is clear from the context; however when it is necessary to be more specific, we shall write it as $\alpha_{[1:M]}$. For a real matrix $\Sigma$, we write its transpose as $\Sigma^t$. The positive semidefinite order is denoted as $\succeq$.  

\section{Correlated Gaussian Sources on a Gaussian Broadcast Channel}

In this section we consider the problem of sending correlated Gaussian sources on a Gaussian broadcast channel, which can be described as follows; see also Fig. \ref{fig:BC} for an illustration. Let the zero-mean Gaussian source be $(S_1[n],S_2[n],\ldots,S_M[n])$ with covariance matrix $\Sigma_{S_1,S_2,\ldots,S_M}$ (or simply $\Sigma_{S_{[1:M]}}$), which is assumed to be full rank.
The channel is given by
\begin{align}
Y_m[n]=X[n]+Z_m[n],\quad m=1,2,\ldots,M,\quad n=1,2,\ldots,N,
\end{align}
where $(Z_1,Z_2,\ldots,Z_M)$ are zero-mean additive noises which are mutually independent, with variances $\sigma^2_{Z_1}\geq \sigma^2_{Z_2}\geq \ldots \geq \sigma^2_{Z_M}$, respectively. Both the sources and the channels noises are independent and identically distributed (i.i.d.) over time. The channel input must satisfy an average power constraint $\frac{1}{N}\sum_{n=1}^N\Ex (X[n]^2)\leq P$. The transmitter encodes the length-$N$ source vector $(S^N_1,S^N_2,\ldots,S^N_M)$ into a length-$N$ channel vector $X^N$, and the $m$-th receiver reconstructs from the channel output vector $Y^N_m$ the source vector ${S}^N_m$ as $\hat{S}^N_m$, resulting in a distortion $D_m=\frac{1}{N}\sum_{n=1}^N\Ex(S_m[n]-\hat{S}_m[n])^2$. We omit a formal problem definition using generic encoding and decoding functions here, which is standard and can be obtained by extending that in, for example, \cite{Tian:11}.

The uncoded scheme of interest has the form
\begin{align}
X[n]=\sum_{m=1}^M \alpha_m S_m[n],\quad n=1,2,\ldots,N,
\end{align}
such that
\begin{align}
\Ex(X^2[n])=P,\quad n=1,2,\ldots,N.
\end{align}
In other words, at each time instance, the channel input is simply a linear combination of the source components with the coefficients $(\alpha_1,\alpha_2,\ldots,\alpha_M)$, such that the resulting signal has a variance that is equal to the power constraint $P$.
We shall assume $\alpha_m\neq 0$ for $m=1,2,\ldots,M$. The decoders simply estimate $S_m[n]$ as $\hat{S}_m[n]=\Ex(S_m[n]|Y_m[n])$, at each time instance $n=1,2,\ldots,N$ at decoder $m=1,2,\ldots,M$. 
Notice that the problem can be equivalently formulated as computation of linear functions of the Gaussian sources on the broadcast channel, however this alternative formulation is notationally more involved.

\begin{figure}[tb]
\centering
\includegraphics[width=13cm]{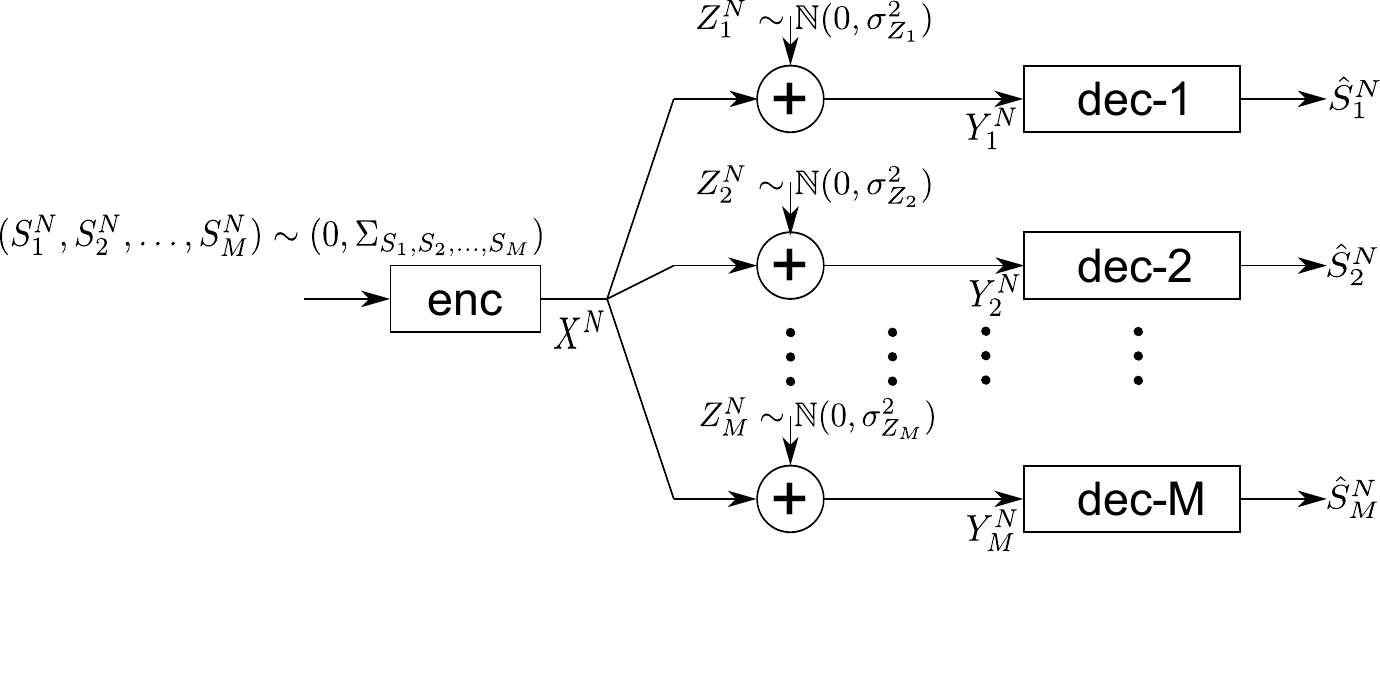}
\vspace{-0.5cm}
\caption{Sending correlated Gaussian sources on a Gaussian broadcast channel.  \label{fig:BC}}
\end{figure}

Define
\begin{align}
\bar{\beta}=(\beta_1,\beta_2,\ldots,\beta_M)^t\triangleq\frac{1}{P}\Sigma_{S_1,S_2,\ldots,S_M}\bar{\alpha}.
\label{eqn:alphatobeta}
\end{align}
The main result on this problem is summarized in the following theorem, which gives a matching condition in a positive semidefinite form.

\begin{theorem}
\label{theorem:BC}
A Gaussian broadcast channel is said to be matched to a given source and the uncoded scheme with non-zero parameters $\bar{\alpha}$, and the distortion vector induced by the given scheme is on the boundary of the achievable distortion region thus optimal, if
\begin{align}
\Sigma^{(0)}\triangleq\Sigma_{V_{[1:M]}} - \Sigma_{S_{[1:M]}}+P\bar{\beta}\bar{\beta}^t
\succeq 0,
\label{eqn:semidefinite}
\end{align}
where the entries of the symmetric  matrix $\Sigma_{V_{[1:M]}}$ are specified as
\begin{align}
\gamma_{j,m}&=-\beta_m\beta_j\frac{P\sigma^2_{Z_m}}{P+\sigma^2_{Z_m}},\quad 1\leq j<m,\quad m=2,3,\ldots,M,\\
\gamma_{m,m}&=\alpha^{-1}_m\left[\beta_m(\sum_{j=1}^{m-1}\alpha_j\beta_j)\frac{P\sigma^2_{Z_m}}{P+\sigma^2_{Z_m}}+\beta_m\sum_{j=m+1}^M\alpha_j\beta_j\frac{P\sigma^2_{Z_j}}{P+\sigma^2_{Z_j}}\right],\quad m=1,2,\ldots,M. \label{eqn:gammamm}
\end{align}
\end{theorem}

This theorem establishes a condition that is sufficient to guarantee a distortion vector induced by the uncoded scheme to be on the boundary of the achievable distortion region, and thus an optimal solution. The matrix $\Sigma_{V_{[1:M]}}$ may seem mysterious at the first sight, however, it will become clear in the proof that it represents the covariance matrix of certain extracted random vectors, whose existence essentially guarantees the optimality of the given uncoded transmission. 

This theorem clearly answers our first question regarding conditions that can be used to certify whether a given uncoded scheme is optimal. In fact, it also provides clues on the second question regarding whether there exist non-trivial channels where such a matching is possible. Indeed, in Section \ref{sec:Cholesky} and Section \ref{sec:BCchannels} we establish several properties of matched channels, through which an answer to the second question is given. Before presenting those results, the proof of this theorem is presented next in two parts: the critical conditions in a novel outer bound are outlined in Section \ref{sec:BCOuter}, and then these conditions for the bound to hold with equality in the uncoded scheme are analyzed in Section \ref{sec:BCinner}. The proof details for the outer bound are relegated to the Appendix.

\subsection{Extracting the Critical Conditions from the Outer Bound}
\label{sec:BCOuter}

In order to obtain the matching condition, we first derive a novel outer bound for this problem. An important technique in the derivation of this outer bound is the introduction of certain appropriate random variables outside of the original problem. This approach is partly motivated by our previous work \cite{TianDiggaviShamai:09,Song:12,Khezeli:14(1),Khezeli:14(2)}, which can further be traced back to Ozarow \cite{Ozarow:80}. Consider $M$ zero-mean Gaussian random variables $(W_1,W_2,\ldots,W_M)$, independent of everything else, with covariance matrix $\Sigma_{W_{[1:M]}}$, and write
\begin{align}
U_m[n]=S_m[n]+W_m[n],\quad m=1,2,\ldots,M,\quad n=1,2,\ldots,N.
\end{align}

The outer bound will be written as a necessary condition that any achievable distortion vector has to satisfy. For this purpose, we bound the following quantity (a summation of entropy powers) for any encoding and decoding functions:
\begin{align}
E(\Sigma_{W_{[1:M]}})\triangleq \sum_{m=1}^M (\sigma^2_{Z_{m}}-\sigma^2_{Z_{m+1}})\exp\left[\frac{2}{N}\sum_{j=1}^mI(U^N_j;Y^N_j|U^N_1,U^N_2,\ldots,U^N_{j-1})\right],
\end{align}
where we have used $\sigma^2_{Z_{M+1}}\triangleq 0$ for notational simplicity. An almost identical quantity was used in \cite{TianDiggaviShamai:09} to obtain an approximate characterization for the distortion region of the Gaussian broadcast problem with bandwidth mismatch. We shall upper-bound this quantity using the channel properties and lower-bound it using the source reconstruction requirements, then combine them to obtain an eventual outer bound.

This quantity can be upper-bounded as given in Appendix \ref{appendix:theorem1} as
\begin{align}
E(\Sigma_{W_{[1:M]}})\leq P+\sigma^2_{Z_1},\label{eqn:BCchannel}
\end{align}
with equality holds if and only if
\begin{align}
h(Y^N_M|S^N_1,S^N_2,\ldots,S^N_M)&=h(Y^N_M|U^N_1,U^N_2,\ldots,U^N_M),\label{eqn:degerate}\\
h(Y^N_1)&=\frac{N}{2}\log 2\pi e( P+\sigma^2_{Z_1}),\label{eqn:BCpower}
\end{align}
and the following condition stemming from the entropy power inequality \cite{CoverThomas} holds with equality
\begin{align}
&\exp\left[\frac{2}{N}h(Y^N_m|U^N_m,U^N_{m-1},\ldots,U^N_1)\right]=\exp\left[\frac{2}{N}h(Y^N_{m+1}|U^N_m,U^N_{m-1},\ldots,U^N_1)\right]+2\pi e[\sigma^2_{Z_m}-\sigma^2_{Z_{m+1}}],\nonumber\\
&\qquad\qquad\qquad\qquad\qquad \qquad\qquad\qquad\qquad\qquad \qquad\qquad\qquad\qquad\qquad m=1,2,\ldots,M.\label{eqn:epi}
\end{align}
The conditions in (\ref{eqn:epi}) are standard, as Bergmans \cite{Bergmans:74} also used the entropy power inequality to establish the Gaussian broadcast channel capacity, and in general a Gaussian codebook suffices to make them equalities. The condition (\ref{eqn:BCpower}) intuitively requires that the power is fully utilized.
The condition (\ref{eqn:degerate}) is however rather peculiar, which essentially requires the noisy source $(U^N_1,U^N_2,\ldots,U^N_M)$ to be as useful as the real source $(S^N_1,S^N_2,\ldots,S^N_M)$ in determining the channel output vector $Y^N_M$.

The quantity $E(\Sigma_{W_{[1:M]}})$ can also be lower-bounded as given in the Appendix, where its individual summands are bounded as
\begin{align}
\exp\left[\frac{2}{N}\sum_{j=1}^mI(U^N_j;Y^N_j|U^N_1,U^N_2,\ldots,U^N_{j-1})\right]\geq \frac{|\Sigma_{S_{[1:m]}}+\Sigma_{W_{[1:m]}}|}{\Pi^m_{j=1}(D_j+\sigma^2_{W_j})},\label{eqn:BCsource}
\end{align}
with equality holds if and only if
\begin{align}
h(U^N_m|Y^N_m)&=\frac{N}{2}\log [2\pi e(D_m+\sigma^2_{W_m} )],\quad m=1,2\ldots,M,\label{eqn:BCcondition1}\\
h(U^N_m|Y^N_m,U^N_1,U^N_2,\ldots,U^N_{m-1})&=h(U^N_m|Y^N_m),\quad m=2,3,\ldots,M.\label{eqn:BCcondition2}
\end{align}
The conditions in (\ref{eqn:BCcondition1}) are standard which can be viewed as requiring the codes to achieve the given distortions with equality, however the conditions in (\ref{eqn:BCcondition2}) are peculiar which essentially require all the information $(Y^N_m,U^N_1,U^N_2,\ldots,U^N_{m-1})$ on $U^N_m$ to be from $Y^N_m$.

Combining (\ref{eqn:BCchannel}) and (\ref{eqn:BCsource}),  we obtain the following result.
\begin{prop}
Any achievable distortion vector $(D_1,D_2,\ldots,D_m)$ must satisfy
\begin{align}
\sum_{m=1}^M (\sigma^2_{Z_{m}}-\sigma^2_{Z_{m+1}})\frac{|\Sigma_{S_{[1:m]}}+\Sigma_{W_{[1:m]}}|}{\Pi^m_{j=1}(D_j+\sigma^2_{W_j})}\leq P+\sigma^2_{Z_1}
\end{align}
for any positive semidefinite $\Sigma_{W_{[1:m]}}$. Moreover, a distortion vector that satisfies (\ref{eqn:degerate}), (\ref{eqn:BCpower}), (\ref{eqn:epi}), (\ref{eqn:BCcondition1}) and (\ref{eqn:BCcondition2}) for some positive semidefinite $\Sigma_{W_{[1:M]}}$ is Pareto-optimal.
\end{prop}

We emphasize that in the approach we shall take, the precise form of this outer bound is less important than the extracted matching conditions (\ref{eqn:degerate}), (\ref{eqn:BCpower}), (\ref{eqn:epi}), (\ref{eqn:BCcondition1}) and (\ref{eqn:BCcondition2}). In fact, the conditions (\ref{eqn:BCpower}), (\ref{eqn:epi}) and (\ref{eqn:BCcondition1}) can be satisfied simply by choosing a jointly Gaussian codebook adjusted linearly to utilize full power, and thus the conditions (\ref{eqn:degerate}) and (\ref{eqn:BCcondition2}) are the only effectual non-trivial conditions. Note that from the problem setting and taking into consideration the fact that physical degradedness is equivalent to stochastic degradedness in the broadcast setting, we have the Markov string
\begin{align}
&Y^N_1\leftrightarrow Y^N_2 \leftrightarrow\ldots\leftrightarrow Y^N_M\leftrightarrow X^N \leftrightarrow (S^N_1,S^N_2,\ldots,S^N_M)\nonumber\\
&\qquad\leftrightarrow (U^N_1,U^N_2,\ldots,U^N_M)\leftrightarrow (U^N_1,U^N_2\leftrightarrow\ldots\leftrightarrow U^N_{M-1})\leftrightarrow\ldots\leftrightarrow U^N_1.\label{eqn:markovstring}
\end{align}
This Markov string is however not sufficient to guarantee (\ref{eqn:degerate}) and (\ref{eqn:BCcondition2}), and thus they require special attention.

\subsection{The Forward Matching Condition}
\label{sec:BCinner}

We first introduce some additional notation and make a few observations. Notice that due to the power constraint, the coefficient vector  $\bar{\alpha}\triangleq (\alpha_1,\alpha_2,\ldots,\alpha_M)^t$ should satisfy
\begin{align}
\bar{\alpha}^t \Sigma_{S_{[1:M]}}\bar{\alpha}=P,
\end{align}
and it follows that
\begin{align}
\sum_{m=1}^M\alpha_m\beta_m=1.
\end{align}

Due to the jointly Gaussian distribution in the uncoded scheme, we can write
\begin{align}
U_m=\beta_m X+ (S_m-\beta_m X)+W_m,\quad m=1,2,\ldots,M,
\end{align}
where the three components are mutually independent, since $\beta_m X=\Ex[S_m|X]$; we have also omitted the time index $[n]$ to simplify the notation. It follows that the covariance matrix of $(U_1,U_2,\ldots,U_M)$ given $Y_m$ can be decomposed as follows
\begin{align}
\Sigma_{U_{[1:M]}|Y_m}=& \sigma^2_{X|Y_m}\bar{\beta}\bar{\beta}^t
+\Sigma_{S_{[1:M]}|X}+\Sigma_{W_{[1:M]}},
\end{align}
where
\begin{align}
\sigma^2_{X|Y_m}=\frac{P\sigma^2_{Z_m}}{P+\sigma^2_{Z_m}},\qquad m=1,2,\ldots,M.
\end{align}
Let $V_m\triangleq(S_m-\beta_m X)+W_m$ for $m=1,2,\ldots,M$, and as a consequence the covariance of the vector $V_{[1:M]}$ is $\Sigma_{S_{[1:M]}|X}+\Sigma_{W_{[1:M]}}$ 

With the above observations, we now return to the derivation of the forward matching condition. As mentioned earlier, we need to substitute the random vectors specified by the uncoded scheme, {\em i.e.,} assigning $X[n]=\sum_{m=1}^M \alpha_m S_m[n]$,  into the critical conditions  (\ref{eqn:degerate}), (\ref{eqn:BCpower}), (\ref{eqn:epi}), (\ref{eqn:BCcondition1}) and (\ref{eqn:BCcondition2}) in order to identify the matching condition. It is straightforward to see that (\ref{eqn:BCpower}), (\ref{eqn:epi}) and (\ref{eqn:BCcondition1}) indeed hold with equality due to the jointly Gaussian distribution of the uncoded scheme, and the chosen coefficients. Thus we only need to focus on (\ref{eqn:degerate}) and (\ref{eqn:BCcondition2}),
which in the context of the uncoded scheme are equivalent to the following conditions in a single-letter form
\begin{align}
h(Z_M)&=h(Y_M|U_1,U_2,\ldots,U_M),\label{eqn:BCasifnonoise}\\
h(U_m|Y_m,U_1,U_2,\ldots,U_{m-1})&=h(U_m|Y_m),\quad m=2,3,\ldots,M.\label{eqn:BCindep}
\end{align}

To satisfy the condition (\ref{eqn:BCindep}) with the jointly Gaussian uncoded scheme,  for any $m=2,3,\ldots,M$, we must have  $\Ex[V_mV_j]+\beta_m\beta_j\sigma^2_{X|Y_m}=0$ for $j=1,2,\ldots,m-1$. This specifies all the off-diagonal terms of $\Sigma_{V_{[1:M]}}$, as
\begin{align}
\Ex[V_mV_j]=\gamma_{m,j}=-\beta_m\beta_j\sigma^2_{X|Y_m},\quad 1\leq j<m,\quad m=2,3,\ldots,M. \label{eqn:crossterms}
\end{align}
It remains to determine the diagonal entries of $\Sigma_{V_{[1:M]}}$.

Notice 
\begin{align}
\sum_{m=1}^M\alpha_m S_m=X=\sum_{m=1}^M\alpha_m\beta_m X
\end{align}
implies that
\begin{align}
\sum_{m=1}^M \alpha_m U_m=X+\sum_{m=1}^M\alpha_m W_m. 
\end{align}
Due to the joint Gaussian distribution and the Markov string $Y_M\leftrightarrow X\leftrightarrow (U_1,U_2,\ldots,U_M)$, in order to satisfy the condition (\ref{eqn:BCasifnonoise}) with equality, we must be able to write $X$ as a linear combination of $(U_1,U_2,\ldots,U_M)$, denoted as $\bar{\alpha}'$.  This implies that
\begin{align}
\sum_{m=1}^M{\alpha}'_m(S_m+W_m)=X=\sum_{m=1}^M\alpha_m S_m,
\end{align}
but this further implies that $\bar{\alpha}'=\bar{\alpha}$, because of the assumption that $\Sigma_{S_{[1:M]}}$ is full rank, and $S_{[1:M]}$ and $W_{[1:M]}$ are independent. It follow that 
\begin{align}
\sum_{m=1}^M \alpha_mV_m=\sum_{m=1}^M \alpha_mW_m=0. \label{eqn:sumzero}
\end{align}
Thus for any $m=1,2,\ldots,M$,
\begin{align}
\sum_{j=1}^M\alpha_j\Ex[V_mV_j]=\Ex[V_m\sum_{j=1}^M\alpha_jV_j]=0.
\end{align}
It follows that $\gamma_{m,m}=\sigma^2_{V_m}$ can be determined from
\begin{align}
\alpha_m\sigma^2_{V_m}=-\sum_{j=1}^{m-1}\alpha_j\Ex[V_mV_j]-\sum_{j=m+1}^M\alpha_j\Ex[V_mV_j]=\beta_m\sigma^2_{X|Y_m}\sum_{j=1}^{m-1}\alpha_j\beta_j+\beta_m\sum_{j=m+1}^M\alpha_j\beta_j\sigma^2_{X|Y_j},
\end{align}
since $\alpha_m\neq 0$.

Thus the conditions (\ref{eqn:BCasifnonoise}) and (\ref{eqn:BCindep}) being equalities uniquely specify the matrix $\Sigma_{V_{[1:M]}}$. Conversely, as long as the matrix $\Sigma^{(0)}$ is positive semidefinite, the conditions  (\ref{eqn:BCasifnonoise}) and (\ref{eqn:BCindep}) hold with equality and the corresponding auxiliary random variables $(W_1,W_2,\ldots,W_M)$ can be found, and the outer bound derived previously is thus tight. This is exactly the matching condition given in Theorem \ref{theorem:BC}.

\textbf{Remark: }The outer bound conditions (\ref{eqn:degerate}) and (\ref{eqn:BCcondition2}) in the context of the uncoded scheme provide two constraints on the matrix $\Sigma_{V_{[1:M]}}$. Their effects on the matrix $\Sigma_{V_{[1:M]}}$ are largely decoupled: the condition required by (\ref{eqn:BCcondition2}) being equal determines the off-diagonal entries of $\Sigma_{V_{1:M]}}$, while the condition (\ref{eqn:degerate}) determines its diagonal entries. This decoupling effect is particularly helpful in deriving the matching condition. In the second problem we consider in the next section, {\em i.e.,} the multiple access channel problem, this decoupling effect is even more pronounced.

\subsection{Cholesky Factorization and a Necessary Condition for Matching}
\label{sec:Cholesky}

The condition given in Theorem \ref{theorem:BC} is in a positive semidefinite form, however, due to the specific problem structure, it can also be represented as a set of recursive conditions, which is discussed in this section. This alternative representation also leads to a necessary condition for matching to hold, which plays an instrumental role for several results given in Section \ref{sec:BCchannels}, where we answer the second question regarding the existence of a non-trivial set of matched channels.

Determining whether a matrix is positive semidefinite is equivalent to computing the LDL decomposition, and checking whether the entries of the resultant diagonal matrix in the decomposition are all non-negative; {\em i.e.}, the matrix $\Sigma^{(0)}$ is positive semidefinite if and only if the diagonal matrix in the LDL decomposition has only non-negative entries. Computationally this can be accomplished with the Cholesky factorization \cite{matrixcomputation} on the matrix $\Sigma^{(0)}$. Here we provide an intuitive description of the Cholesky factorization in the context of the problem being considered, and its conceptual interpretation as the recursive thresholding determination for the channel to yield a matching. Readers more interested in the precise mathematical derivation can skip to the proof of Lemma \ref{lemma:cholesky} directly.

In the first step of the Cholesky factorization, we use symmetric column and row Gaussian elimination to eliminate all the entries of the $M$-th column and the $M$-th row, except the diagonal entry\footnote{Strictly speaking, this yields a decomposition with an upper triangular matrix instead of a lower triangular one.}. Denote the resulting upper-left $(M-1)\times(M-1)$ matrix after this first  step as $\Sigma^{(1)}$. A necessary condition for the matrix $\Sigma^{(0)}$ to be positive definite is that the lower right entry of the matrix $\Sigma^{(0)}$ is strictly positive, or all the entries on the last column are zero. Notice that the condition only involves $\sigma^2_{X|Y_M}$, or equivalently only the channel noise power $\sigma^2_{Z_M}$, which yields a necessary condition on $\sigma^2_{Z_M}$ in the form of $\sigma^2_{Z_M}\geq f^{(0)}(P,\bar{\alpha})$.

Continuing the Cholesky factorization on $\Sigma^{(1)}$,  a similar necessary condition is thus its lower right entry is strictly positive,
or the entries on the $(M-1)$-th row of  $\Sigma^{(1)}$ are zero. Similarly as the previous step, the condition on $\sigma^2_{Z_{M-1}}$ is found to be in the form that $\sigma^2_{Z_{M-1}}\geq  f^{(1)}(P,\bar{\alpha},\sigma^2_{Z_M})$.

Continuing this process will yield a set of conditions in the form of
\begin{align}
\sigma^2_{Z_{m}}\geq  f^{(M-m)}(P,\bar{\alpha},\sigma^2_{Z_M},\sigma^2_{Z_M},\sigma^2_{Z_{M-1}},\ldots,\sigma^2_{Z_{m+1}}),\quad m=M,M-1,\ldots,1.
\end{align}
The matrix $\Sigma^{(0)}$ is positive semidefinite if and only if all such threshold conditions are satisfied.

Notice that the threshold function  $f^{(M-m)}(P,\bar{\alpha},\sigma^2_{Z_M},\sigma^2_{Z_M},\sigma^2_{Z_{m-1}})$ for $\sigma^2_{Z_{m}}$ depends on the channel noise power values $(\sigma^2_{Z_M},\sigma^2_{Z_{M-1}},\ldots,\sigma^2_{Z_{m+1}})$, but not on $(\sigma^2_{Z_{1}},\sigma^2_{Z_{2}},\ldots,\sigma^2_{Z_{m-1}})$. Thus these functions $f^{(m)}(\cdot)$, $m=M,M-1,\ldots,1$ can be viewed as a recursive threshold checking (or determination) procedure, and the channel noise power $\sigma^2_{Z_m}$ needs to be chosen to be larger than the threshold determined by $(\sigma^2_{Z_M},\sigma^2_{Z_{M-1}},\ldots,\sigma^2_{Z_{m+1}})$ in every step to yield a matching. Given the above observation, it is natural to speculate that if a channel is matched, then any more noisy channel also induces a match. This intuition is in fact correct, and the statement is made more rigorous in the next section as Corollary \ref{coro1}. This behavior is reminiscent of the optimality of broadcasting a single Gaussian source on a bandwidth-matched Gaussian channel, and has also been previously observed for broadcast bivariate Gaussian sources \cite{Lapidoth:10BC}.


We can thus apply the Cholesky factorization technique on the matrix $\Sigma_{V_{[1:M]}}$ to obtain a necessary condition for matching to exist.

\begin{lemma}
\label{lemma:cholesky}
For the matrix  $\Sigma_{V_{[1:M]}}$ constructed previously to be positive semidefinite (with $\sigma^2_{Z_M}>0$), it must be true that $\alpha_i\beta_i\geq 0$, $i=1,2,\ldots,M$.
\end{lemma}

Note that this condition is essentially independent of the channel, as long as the channel is not perfect. This lemma is proved in Appendix \ref{appendix:lemma1}.

\subsection{Properties and Existence of Matched Channels}
\label{sec:BCchannels}

With Lemma 1, we can establish several properties of the set of matched channels, given next as corollaries to Theorem \ref{theorem:BC}. Their proofs are provided in Appendix \ref{appendix:coro1}-\ref{appendix:coro3}. These properties  essentially provide an answer to the second question posed earlier, and we shall further illustrate such sources and channels using an example.

\begin{corollary}
\label{coro1}
If the uncoded scheme is matched on a broadcast channel with noise powers given as $(\sigma^2_{Z_1},\sigma^2_{Z_2},\ldots,\sigma^2_{Z_M})$, then it is matched and thus optimal on any channel with noise powers $\sigma^2_{Z^+_1}\geq \sigma^2_{Z^+_2}\geq \ldots\geq \sigma^2_{Z^+_M}$ where $\sigma^2_{Z^+_m}\geq \sigma^2_{Z_m}$, $m=1,2,\ldots,M$.
\end{corollary}

The corollary reveals a property of matched channels: once a channel is matched, any channel with more noise is also a matched channel and thus the uncoded scheme is optimal. The next corollary states, from the perspective of only the source and the uncoded scheme parameters, a necessary and sufficient condition for matching to exist.

\begin{corollary}
\label{coro2}
Matching (on some broadcast channels with finite noise powers) exists, if and only if $\alpha_i\beta_i>0$ and the matrix $\Pi\Sigma_{S_{[1:M]}}\Pi$ has its largest eigenvalue being 1 with multiplicity 1, where $\Pi$ is a diagonal matrix with diagonal entries being
\begin{align}
\left(\frac{\alpha_1}{\sqrt{\alpha_1\sum_{i=1}^M\rho_{1,i}\alpha_i}},\frac{\alpha_2}{\sqrt{\alpha_2\sum_{i=1}^M\rho_{2,i}\alpha_i}},\ldots,\frac{\alpha_M}{\sqrt{\alpha_M\sum_{i=1}^M\rho_{M,i}\alpha_i}}\right),
\end{align}
where $\rho_{i,j}$ is used to denote the entries of $\Sigma_{S_{[1:M]}}$. Moreover, if the above condition holds, then any channel with $\sigma^2_{Z_1}\geq \sigma^2_{Z_2}\ldots \sigma^2_{Z_M}\geq \sigma^2_Z$ is a matched channel, where  $\sigma^2_Z=\frac{\lambda_2P}{1-\lambda_2}$, and $\lambda_2$ is the second largest eigenvalue of the matrix $\Pi\Sigma_{S_{[1:M]}}\Pi$.
\end{corollary}


\textbf{Remark:} It should be noted that the condition in the first part of the corollary is the most general condition that can be derived using Theorem \ref{theorem:BC}, but it does necessarily capture all the cases that an analog scheme can be optimal, which stems from the fact that the outer bound we derived may not be tight.

\textbf{Remark:} If the entries of $\mbox{diag}(\bar{\alpha})\Sigma_{S_{[1:M]}}\mbox{diag}(\bar{\alpha})$ are strictly positive, then matching is always possible.
This  follows from the fact that the matrix $\Pi\Sigma_{S_{[1:M]}}\Pi$ has positive entries, and $\bar{v}^t_1=(\sqrt{\alpha_1\beta_1}, \sqrt{\alpha_2\beta_2}, ...,   \sqrt{\alpha_M\beta_M})$ is its positive eigenvector, such that $1$ is its largest eigenvalue with multiplicity 1 (by Perron-Frobenius Theorem \cite{matrixanalysis}).

Different from the case discussed in the previous remark, the next corollary gives another sufficient condition for matching to occur when the sources and the coding parameters satisfy the same positive correlation condition.

\begin{corollary}
\label{coro3}
Let the entries of the matrix $\mbox{diag}(\bar{\alpha})\Sigma_{S_{[1:M]}}\mbox{diag}(\bar{\alpha})$ be strictly positive. Define
\begin{align}
\sigma^2_{Z^*_m}\triangleq \max_{j<m} \frac{\beta_j\beta_m}{\rho_{j,m}}P^2-P, \quad m=2,3,\ldots,M.
\end{align}
Any channel with $\sigma^2_{Z_1}\geq \sigma^2_{Z_2}\geq \ldots\geq \sigma^2_{Z_M}$ such that $\sigma^2_{Z_m}\geq \sigma^2_{Z^*_m}$ for $m=2,3,\ldots,M$ is a matched channel.
\end{corollary}

\textbf{Remark:} $\sigma^2_{Z^*_m}$ as defined above may in fact be negative for some $m$. However this does not cause any discrepancy, because of the existence of the additional requirement $\sigma^2_{Z_1}\geq \sigma^2_{Z_2}\geq \ldots\geq \sigma^2_{Z_M}$. As a sanity check, notice that
\begin{align}
\sum_{i=1}^M(\alpha_i)(\beta_i\beta_m P)=P\beta_m=\sum_{i=1}^M(\alpha_i)(\rho_{i,m}),
\end{align}
but $\beta_M^2P<\rho_{M,M}=\sigma^2_{S_M}$ unless $\alpha_1=\alpha_2=\ldots=\alpha_{M-1}=0$, which however would contradict our assumption.
It thus follows
\begin{align}
\max_{j<M} \frac{\beta_j\beta_mP}{\rho_{j,M}}> 1,
\end{align}
and thus $\sigma^2_{Z^*_M}>0$ always holds under the condition in the corollary.

\textbf{Remark:} For the symmetric case where $\sigma^2_{S_i}=\sigma^2_S$, $\sigma^2_{S_iS_j}=\rho\sigma^2_S$, $\alpha_i=\alpha$ and $\Ex[S_i|X]=\beta X$, for $i=1,2,\ldots,M$. A necessary and sufficient condition for matching is simply
\begin{align}
\frac{\sigma^2_{Z_M}}{P+\sigma^2_{Z_M}}\geq \frac{1-\rho}{1+(M-1)\rho}.
\label{eqn:symmetric}
\end{align}

To see this, notice that
\begin{align}
P=\alpha^2M\sigma^2_S[1+(M-1)\rho].
\end{align}
and $\beta$ can be computed as
\begin{align}
\beta =\frac{1}{\alpha M}.
\end{align}
Checking the first condition in the Cholesky factorization, it is easily verified that (\ref{eqn:symmetric}) is a necessary condition for matching. However, from Corollary \ref{coro3}, it is seen that it is sufficient to choose any $\sigma^2_{Z_m}\geq \sigma^2_{Z^*_m}$, where
$\sigma^2_{Z^*_m}=\frac{\beta^2}{\rho \sigma^2_{S}}P^2-P$, $m=2,3,\ldots,M$. This is exactly condition (\ref{eqn:symmetric}).

\subsection{An Example:  A Source with Three Components}

\begin{figure}[tb]
\centering
\includegraphics[width=8cm]{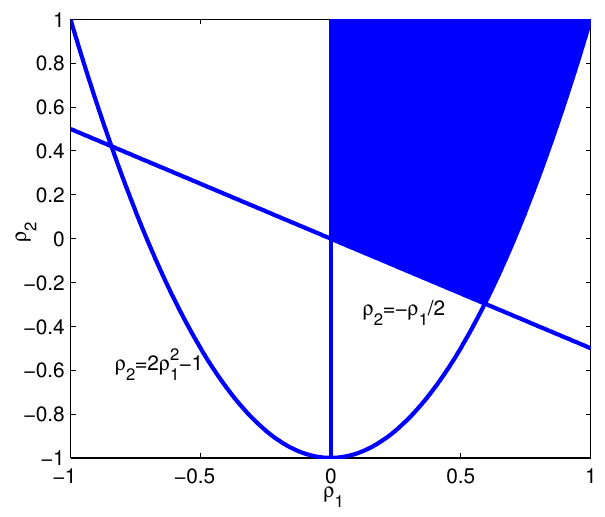}
\caption{Determining the $(\rho_1,\rho_2)$ pairs for which matching is possible, given in shade.  \label{fig:rho1rho2}}
\end{figure}

Let us consider a source with three components whose covariance matrix is either
\begin{align}
\Sigma_{S_1,S_2,S_3}=\begin{pmatrix}
1& \rho_{1}&\rho_{2}\\
\rho_{1}&1&\rho_{1}\\
\rho_{2}&\rho_{1}&1
\end{pmatrix},\label{eqn:cov1}
\end{align}
or
\begin{align}
\Sigma_{S_1,S_2,S_3}=\begin{pmatrix}
1& \rho_{2}&\rho_{1}\\
\rho_{2}&1&\rho_{1}\\
\rho_{1}&\rho_{1}&1
\end{pmatrix},\label{eqn:cov2}
\end{align}
and further assume that the coefficients are chosen as $\alpha_1=\alpha_2=\alpha_3=1$ in the uncoded scheme.
In addition to the constraint that the matrix $\Sigma_{S_1,S_2,S_3}$ must be positive definite, for a matching to exist, the condition in Corollary \ref{coro2} must be satisfied.
 It can be shown that the eigenvalues of $\Pi\Sigma_{S_1,S_2,S_3}\Pi$ are
\begin{align}
\lambda_1=1,\quad \lambda_2=\frac{-2\rho_1^2 + \rho_2 + 1}{2\rho_1^2+2\rho_1\rho_2+3\rho_1 + \rho_2  + 1},\quad \lambda_3=\frac{1-\rho_2}{\rho_1 + \rho_2 + 1},
\end{align}
and we must have $\lambda_2<1$ and $\lambda_3<1$. In Appendix \ref{appendix:source2}, we show that the valid choices are the $(\rho_1,\rho_2)$ pairs such that
\begin{align}
\rho_2< 1,\quad 0< \rho_1< 1,\quad \rho_1+2\rho_2>0,\quad \rho_2> 2\rho_1^2-1. \label{eqn:source2}
\end{align}
The corresponding region is plotted in Fig \ref{fig:rho1rho2}. Notice that the two matrices are equivalent for the purpose of determining whether matching is possible, thus the region in Fig \ref{fig:rho1rho2} is valid for both cases.

Next let us fix a $(\rho_1,\rho_2)$ pair, and consider the region of $(\frac{P\sigma^2_{Z_2}}{P+\sigma^2_{Z_2}},\frac{P\sigma^2_{Z_3}}{P+\sigma^2_{Z_3}})$ pairs such that matching occurs. The tradeoffs can be computed explicitly, and are illustrated in Fig. \ref{fig:N1N2} for $(\rho_1,\rho_2)=(\frac{1}{2},\frac{1}{6})$. The circles in the plots give the channels specified by Corollary \ref{coro2}. The channels given by Corollary \ref{coro3} can be computed directly (given as the dots), which is loose in the first case, but on the lower boundary (and it is an extreme point) for the second case. Since $\sigma^2_{Z_3}\geq \sigma^2_{Z_2}$, we also include this boundary in the plot. For the first case, the boundary $\frac{P\sigma^2_{Z_3}}{P+\sigma^2_{Z_3}}<P$ is also shown, while for the second, the lower bound $y\geq\frac{16}{15}$  required by the function $f^{(0)}(P,\bar{\alpha})$ in the first step of the Cholesky factorization is shown.
The corresponding channels that matching occurs are those inside the \lq\lq{}fan\rq\rq{} regions. Note that there is a tension between the noise powers $\sigma^2_{Z_2}$ and $\sigma^2_{Z_3}$ for matching to occur with the fixed source and uncoded scheme parameters.

\begin{figure}[tb]
\centering
\includegraphics[width=16cm]{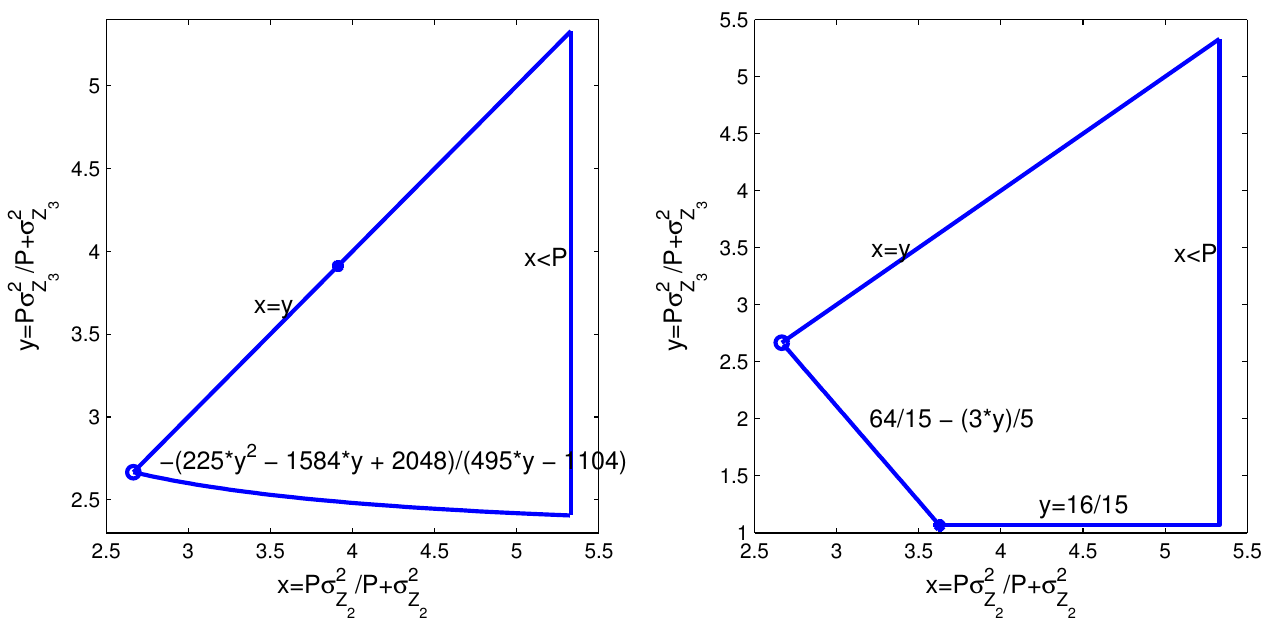}
\caption{Illustration of the regions of matched channel parameters when $(\rho_1,\rho_2)=(\frac{1}{2},\frac{1}{6})$ for the two covariance matrices (\ref{eqn:cov1}) and (\ref{eqn:cov2}), respectively.  \label{fig:N1N2}}
\end{figure}

\section{Vector Gaussian CEO on a Gaussian Multiple-Access Channel}

In this section we consider the problem of sending correlated Gaussian sources on a Gaussian multiple-access channel, where the transmitters observe noise linear combinations of the source components; see also Fig. \ref{fig:MAC} for an illustration. 

A zero-mean vector Gaussian source $(S_1[n],S_2[n],\ldots,S_M[n])$ has a covariance matrix $\Sigma_{S_1,S_2,\ldots,S_M}$ (or simply $\Sigma_{S_{[1:M]}}$). There are a total of $L$ sensors, whose observations are $(T_1[n],T_2[n],\ldots,T_L[n])$, respectively, with covariance matrix $\Sigma_{T_1,T_2,\ldots,T_L}$ (or simply $\Sigma_{T_{[1:L]}}$).  The source and observations are jointly Gaussian. Each sensor observes $T^N_\ell$, encodes it under an average transmission power constraint $P_\ell$, $\ell=1,2,\ldots,L$.
The channel output is given as
\begin{align}
Y[n]=Z[n]+\sum_{\ell=1}^L\delta_\ell X_\ell[n], \quad n=1,2,\ldots, N,
\end{align}
where the channel amplification factors $\delta_\ell>0$, $\ell=1,2,\ldots,L$. The receiver wishes to reconstruct $(S^N_1,S^N_2,\ldots,S^N_M)$ using channel output $Y^N$ to minimize the individual MSE measure, which achieves MSE distortion $D_m$ for $S_m$, {\em i.e.}, $D_m=\frac{1}{N}\sum_{n=1}^N\Ex(S_m[n]-\hat{S}_m[n])^2$.

Notice that due to the jointly Gaussian distribution, we can write
\begin{align}
 \tilde{S}_m\triangleq\Ex[S_m|T_1,T_2,\ldots,T_L]=\sum_{\ell=1}^L\gamma_{m,\ell}T_\ell,\quad m=1,2,\ldots,M.\label{eqn:Stilde}
\end{align}
The parameters $\gamma_{m,\ell}$ can be conveniently written as a matrix $\Gamma$, and computed as
\begin{align}
{\Gamma}=\Sigma_{S_{[1:M]},T_{[1:L]}}\Sigma^{-1}_{T_{[1:L]}},
\end{align}
where $\Sigma_{S_{[1:M]},T_{[1:L]}}$ is the cross-covariance matrix between the random vectors $(S_1,S_2,\ldots,S_M)$ and $(T_1,T_2,\ldots,T_L)$.

The problem can be equivalently formulated as computation of linear functions of Gaussian sources on the multiple-access channel. In this alternative setting, the functions to be computed are $(S_1,S_2,\ldots,S_M)$, which can be represented as noisy linear functions of the sensor observations $(T_1,T_2,\ldots,T_L)$. This alternative formulation is notationally more involved in the current problem setting, but we shall explore this connection in a separate work.

We assume $M\leq L$ since the other case can be reduced to this case without loss of generality. We will consider the case that the matrices  $\Sigma_{S_{[1:M,]}}$,  $\Sigma_{T_{[1:L]}}$, $\Sigma_{\tilde{S}_{[1:M]}}$ and $\Sigma_{S_{[1:M]},T_{[1:L]}}$ all have full (row) rank, which hold in general except certain degenerate cases. Denote the entries of $\Sigma_{T_{[1:L]}}$ as $\psi_{i,j}$. The uncoded scheme we consider is
\begin{align}
X_\ell[n]=\eta_\ell\sqrt{\frac{P_\ell}{\psi_{\ell,\ell}}}T_\ell[n],\quad \ell=1,2,\ldots,L,\quad n=1,2,\ldots,N,\label{eqn:Xform}
\end{align}
where $\eta_\ell$ is either $+1$ or $-1$  to be specified next. In other words, each sensor sends its noisy observations directly using the full power, but it can choose whether to negate its observations. The $m$-th receiver estimates $S_m[n]$ as $\hat{S}_m[n]=\Ex[S_m[n]|Y[n]]$.

Define
\begin{align}
&\bar{\alpha}\triangleq\left[\Sigma_{S_{[1:M]},T_{[1:L]}}\Sigma^t_{S_{[1:M]},T_{[1:L]}}\right]^{-1}\nonumber\\
&\qquad\cdot\Sigma_{S_{[1:M]},T_{[1:L]}}\Sigma_{T_{[1:L]}}\left(\delta_1 \eta_1\sqrt{\frac{P_1}{\psi_{1,1}}},\delta_2 \eta_2\sqrt{\frac{P_2}{\psi_{2,2}}},\ldots,\delta_L \eta_L\sqrt{\frac{P_L}{\psi_{L,L}}}\right)^t,\label{eqn:alphas}
\end{align}
and we assume $\alpha_m\neq 0$, $m=1,2,\ldots,M$, which is true in general except certain degenerate cases. Our main result on this problem is summarized in the following theorem.

\begin{figure}[tb]
\centering
\includegraphics[width=16cm]{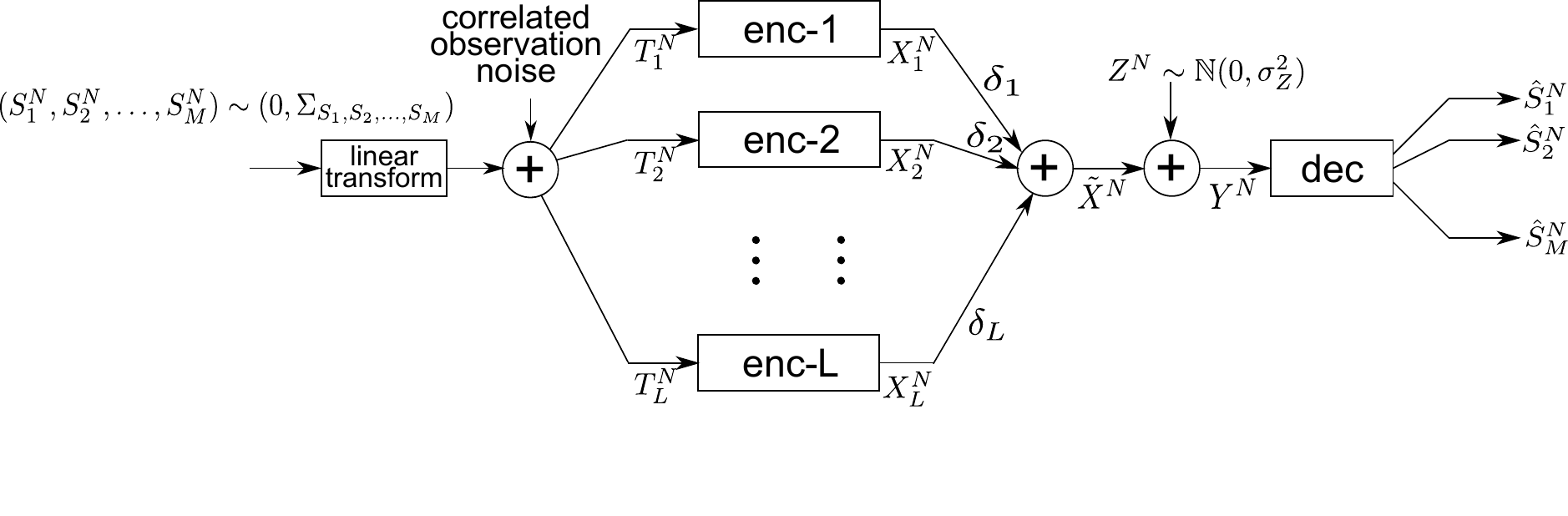}
\vspace{-0.5cm}
\caption{Sending correlated Gaussian sources on a Gaussian multiple-access channel with noisy observations.  \label{fig:MAC}}
\end{figure}

\begin{theorem}
\label{theorem:MAC}
A Gaussian multiple-access channel is said to be matched to a given Gaussian source and the uncoded scheme with parameters $\bar{\eta}$, and the distortion vector induced by the given scheme is on the boundary of the achievable distortion region and thus optimal, if
\begin{enumerate}
\item $\eta_\ell\eta_{\ell'}\psi_{\ell,\ell'}\geq 0,\quad 1\leq \ell<\ell'\leq L$;
\item The vector $\left(\delta_1 \eta_1\sqrt{\frac{P_1}{\psi_{1,1}}},\delta_2 \eta_2\sqrt{\frac{P_2}{\psi_{2,2}}},\ldots,\delta_L \eta_L\sqrt{\frac{P_L}{\psi_{L,L}}}\right)\Sigma_{T_{[1:L]}}$ is in the row space of the matrix $\Sigma_{S_{[1:M]},T_{[1:L]}}$;
\item $\sigma^2_Z\geq \frac{\lambda_2P}{1-\lambda_2}$, where $\lambda_2$ is the second largest eigenvalue of the matrix $\Pi\Sigma_{\tilde{S}_{[1:M]}}\Pi$, 
\begin{align}
P\triangleq \sum_{\ell=1}^L\delta^2_\ell P_\ell+2\sum_{\ell=1}^L\sum_{\ell'=\ell+1}^L\rho^*_{\ell,\ell'}\delta_\ell\delta_{\ell'}\sqrt{P_\ell P_{\ell'}},
\end{align}
and $\Pi$ is a diagonal matrix with diagonal entries
\begin{align}
\left(\frac{\alpha_1}{\sqrt{\alpha_1\sum_{i=1}^M\rho_{1,i}\alpha_i}},\frac{\alpha_2}{\sqrt{\alpha_2\sum_{i=1}^M\rho_{2,i}\alpha_i}},\ldots,\frac{\alpha_M}{\sqrt{\alpha_M\sum_{i=1}^M\rho_{M,i}\alpha_i}}\right).
\end{align}
and $\rho_{m,j}$'s are the entries of the matrix
\begin{align}
\Sigma_{\tilde{S}_{[1:M]}}=\Sigma_{S_{[1:M]},T_{[1:L]}}\Sigma^{-1}_{T_{[1:L]}}\Sigma^t_{S_{[1:M]},T_{[1:L]}}.
\end{align}
\end{enumerate}
\end{theorem}

These conditions can be intuitively explained as follows: condition one guarantees that the channel inputs from all transmitters coherently add up; condition two stems from the  requirement  that the noisy observations should serve the same role as the underlying sources for the chosen power constraints and amplification factors, {\em i.e.,} as if the observation noise does not exist; condition three is similar to the effect in the previous problem where once a channel is matched, a more noisy channel will also induce a match.

When all $\psi_{\ell,\ell\rq{}}\geq 0$, we can simply choose $\eta_\ell=+1$ (or $-1$) for all $\ell$ to satisfy the first condition. However, when some of the terms $\psi_{\ell,\ell\rq{}}$ are negative, a simple algorithmic approach can be used to determine whether there exists a valid assignment of $\{\eta_\ell,\ell=1,2,\ldots,L\}$. In fact this condition is completely source dependent, and the choice of $\{\eta_\ell,\ell=1,2,\ldots,L\}$ is unique up to a negation (assuming any component $T_\ell$ is not completely independent of the others), and thus can be considered fixed for a given source observation covariance matrix.

The proof of this theorem also has two parts given in Section \ref{sec:outerMAC} and Section \ref{sec:innerMAC}. This theorem answers the first question regarding the conditions to certify whether the uncoded scheme is optimal in this communication problem. The answer to the second question for this problem turns out to be simpler than that in the broadcast case, and we discuss in Section \ref{sec:MACexample} as special case several problems previously considered in the literature.

\subsection{Extracting the Critical Conditions from the Outer Bound}
\label{sec:outerMAC}

Define 
\begin{align}
\Delta_m\triangleq\Ex(S_m-\tilde{S}_m)^2,\quad m=1,2,\ldots,M,
\end{align}
and thus
\begin{align}
\Ex\tilde{S}^2_m=\sigma^2_{S_m}-\Delta_m,\quad m=1,2,\ldots,M.
\end{align}
In this remote coding setting, in essence $\tilde{S}_m$'s as defined in (\ref{eqn:Stilde}) are the observable portion of the underlying sources.  The overall distortion can thus be decomposed into two independent parts: the first part is due to encoding the observable portion of the underlying sources $\tilde{S}_m$'s,  and the second is due to the inherent noisy nature of the observations which induces a fixed distortion $\Delta_m$. Thus encoding the source $S_m$ to distortion $D_m$ is equivalent to encoding the source $\tilde{S}_m$ to distortion $D_m-\Delta_m$.

We can now derive an outer bound by combining the approach used in the broadcast problem with a technique based on Witsenhausen's bound \cite{Witsenhausen}. Again consider $M$ auxiliary zero-mean Gaussian random variables $(W_1,W_2,\ldots,W_M)$ with covariance matrix $\Sigma_{W_{[1:M]}}$, which are independent of everything else, and write
\begin{align}
U_m[n]=\tilde{S}_m[n]+W_m[n],\quad m=1,2,\ldots,M,\quad n=1,2,\ldots,N.
\end{align}
Notice the Markov string
\begin{align}
(U^N_1,U^N_2,\ldots,U^N_M)\leftrightarrow (\tilde{S}^N_1,\tilde{S}^N_2,\ldots,\tilde{S}^N_M)\leftrightarrow (T^N_1,T^N_2,\ldots,T^N_L)\leftrightarrow (X^N_1,X^N_2,\ldots,X^N_L)\leftrightarrow Y,
\end{align}
and we can write using the data processing inequality \cite{CoverThomas} that
\begin{align}
I(X^N_1,X^N_2,\ldots,X^N_L;Y^N)\geq I(U^N_1,U^N_2,\ldots,U^N_M;Y^N),\label{eqn:dataprocessing}
\end{align}
where equality holds if and only if
\begin{align}
h(Y^N|X^N_1,X^N_2,\ldots,X^N_L)=h(Y^N|U^N_1,U^N_2,\ldots,U^N_M).\label{eqn:MACindep}
\end{align}

Following the exact steps as in \cite{Gastpar:08} (see also \cite{Lapidoth:10MAC}) and applying Witsenhausen's bound \cite{Witsenhausen}, we can obtain
\begin{align}
I(X^N_1,X^N_2,\ldots,X^N_L;Y^N)\leq \frac{N}{2}\log \left(1+\frac{P}{\sigma^2_Z}\right)
\label{eqn:channelcondition}
\end{align}
where $\rho^*_{\ell,\ell'}=|\psi_{\ell,\ell'}(\psi_{\ell,\ell}\psi_{\ell',\ell'})^{-\frac{1}{2}}|$. This inequality intuitively says that the mutual information between the channel inputs and the output is upper bounded by the capacity of a point-to-point channel, whose power constraint is equal to the resultant signal power when all the inputs on the multiple-access channel are coherently added. We will not attempt to further simplify this condition at this point, since in the context of the uncoded scheme, it has a particularly simple form.

The right hand side of (\ref{eqn:dataprocessing}) can be bounded similarly as in the broadcast problem. Here the equivalent source is $({\tilde{S}_1,\tilde{S}_2,\ldots,\tilde{S}_M})$, and the distortion vectors are $(D_1-\Delta_1,D_2-\Delta_2,\ldots,D_M-\Delta_M)$, and moreover, $\sigma^2_{Z_m}=\sigma^2_{Z}$ for $m=1,2,\ldots,M$. We thus arrive at
\begin{align}
I(U^N_1,U^N_2,\ldots,U^N_M;Y^N)\geq\frac{N}{2}\log\frac{|\Sigma_{\tilde{S}_{[1:M]}}+\Sigma_{W_{[1:M]}}|}{\Pi^M_{m=1}(D_m-\Delta_m+\sigma^2_{W_m})},
\label{eqn:sameasBC}
\end{align}
where equality holds if and only if
\begin{align}
h(U^N_m|Y^N)&=\frac{N}{2}\log [2\pi e(D_m-\Delta_m+\sigma^2_{W_m} )],\quad m=1,2,\ldots,M,\label{eqn:MACcondition1}\\
h(U^N_m|Y^N,U^N_1,U^N_2,\ldots,U^N_{m-1})&=h(U^N_m|Y^N),\quad m=2,3,\ldots,M.\label{eqn:MACcondition2}
\end{align}

An outer bound on the achievable distortion is then obtained by combining (\ref{eqn:dataprocessing}), (\ref{eqn:channelcondition}) and (\ref{eqn:sameasBC}), which we summarize below.
\begin{prop}
Any achievable distortion vector $(D_1,D_2,\ldots,D_M)$ must satisfy the inequality
\begin{align}
\frac{|\Sigma_{\tilde{S}_{[1:M]}}+\Sigma_{W_{[1:M]}}|}{\Pi^M_{m=1}(D_m-\Delta_m+\sigma^2_{W_m})}\leq \left(1+\frac{P}{\sigma^2_Z}\right),
\end{align}
for any positive semidefinite $\Sigma_{W_{[1:M]}}$. Moreover, a distortion vector that makes (\ref{eqn:MACindep}), (\ref{eqn:MACcondition1}) and (\ref{eqn:MACcondition2}) hold, and (\ref{eqn:channelcondition})  hold with equality for some positive semidefinite $\Sigma_{W_{[1:M]}}$ is Pareto-optimal.
\end{prop}

We emphasize that for the purpose of this work, the precise form of this outer bound is less important than the extracted matching conditions (\ref{eqn:MACindep}), (\ref{eqn:MACcondition1}) and (\ref{eqn:MACcondition2}), and (\ref{eqn:channelcondition}) being equality. The condition (\ref{eqn:channelcondition}) being equality and the condition (\ref{eqn:MACcondition1}) can be satisfied simply by choosing a jointly Gaussian coding scheme adjusted linearly to utilize the full power, and the conditions (\ref{eqn:MACindep}) and (\ref{eqn:MACcondition2}) are almost identical to  (\ref{eqn:degerate}) and (\ref{eqn:BCcondition2}) in the broadcast case.

\subsection{The Forward Matching Conditions}
\label{sec:innerMAC}

Since the uncoded scheme takes single letter encoding function, (\ref{eqn:channelcondition}) being equality is equivalent to
\begin{align}
I(X_1,X_2,\ldots,X_L;Y)=\frac{1}{2}\log \left(1+\frac{P}{\sigma^2_Z}\right).
\end{align}
Because in the uncoded scheme the channel input $X$ is given in (\ref{eqn:Xform}),  the equality holds as long as
\begin{align}
\eta_\ell\eta_{\ell'}\psi_{\ell,\ell'}\geq 0,\quad 1\leq \ell<\ell'\leq L.
\end{align}
This yields the first condition stated in Theorem \ref{theorem:MAC}.

The conditions (\ref{eqn:MACindep}) and (\ref{eqn:MACcondition2}) in the context of uncoded scheme are equivalent to
\begin{align}
h(Z)&=h(Y|U_1,U_2,\ldots,U_M),\label{eqn:MACasinnonoise}\\
h(U_m|Y,U_1,U_2,\ldots,U_{m-1})&=h(U_m|Y),\quad m=2,3,\ldots,M.\label{eqn:MACcondition2singleletter}
\end{align}

Denote
\begin{align}
\tilde{X}=\sum_{\ell=1}^L\delta_\ell X_\ell.
\end{align}
For (\ref{eqn:MACasinnonoise}) to hold with equality,  two conditions must hold
\begin{align}
\Ex[\tilde{X}|\tilde{S}_1,\tilde{S}_2,\ldots,\tilde{S}_M]=\tilde{X},\label{eqn:dataprocessing1}
\end{align}
and
\begin{align}
\Ex[\tilde{X}|U_1,U_2,\ldots,U_M]=\tilde{X}.\label{eqn:dataprocessing2}
\end{align}

Let us consider the first condition  (\ref{eqn:dataprocessing1}). Due to the jointly Gaussian distribution, there exists a set of coefficients $(\alpha_1,\alpha_2,\ldots,\alpha_M)$ such that
\begin{align}
\Ex[\tilde{X}|\tilde{S}_1,\tilde{S}_2,\ldots,\tilde{S}_M]=\sum_{m=1}^M\alpha_m\tilde{S}_m=\sum_{m=1}^M\alpha_m\sum_{\ell=1}^L\gamma_{m,\ell}T_\ell.
\end{align}
However notice that
\begin{align}
\tilde{X}=\sum_{\ell=1}^L\delta_\ell \eta_\ell\sqrt{\frac{P_\ell}{\psi_{\ell,\ell}}}T_\ell,
\end{align}
thus the condition (\ref{eqn:dataprocessing1}) is equivalent to the fact that the vector
\begin{align}
\left(\delta_1 \eta_1\sqrt{\frac{P_1}{\psi_{1,1}}},\delta_2 \eta_2\sqrt{\frac{P_2}{\psi_{2,2}}},\ldots,\delta_L \eta_L\sqrt{\frac{P_L}{\psi_{L,L}}}\right)
\end{align}
is in the row space of the matrix ${\Gamma}$. Equivalently, the vector
$$\left(\delta_1 \eta_1\sqrt{\frac{P_1}{\psi_{1,1}}},\delta_2 \eta_2\sqrt{\frac{P_2}{\psi_{2,2}}},\ldots,\delta_L \eta_L\sqrt{\frac{P_L}{\psi_{L,L}}}\right)\Sigma_{T_{[1:L]}}$$
needs to be in the row space of the matrix $\Sigma_{S_{[1:M]},T_{[1:L]}}$. This leads to the second condition stated in Theorem \ref{theorem:MAC}.
When this condition is satisfied, the coefficients $\bar{\alpha}$ can be determined exactly as in (\ref{eqn:alphas}).

The conditions (\ref{eqn:MACcondition2singleletter}) and (\ref{eqn:dataprocessing2}) are now identical to the broadcast case with $\tilde{S}_1,\tilde{S}_2,\ldots,\tilde{S}_M$ being the sources and $\tilde{X}$ being the channel input, and all the receivers in a broadcast channel that has the same channel noise variance. By Corollary 2, such a channel is matched when the second largest eigenvalue of the matrix $\Pi\Sigma_{\tilde{S}_{[1:M]}}\Pi$ is less than $\frac{\sigma^2_Z}{P+\sigma^2_Z}$, or in other words, the noise power must be above or equal to the given threshold stated in Theorem \ref{theorem:MAC}.

\textbf{Remark:} The first condition in Theorem \ref{theorem:MAC} generally has a unique solution if it can be satisfied, up to a negation of all the signs of the channel input signals. The second condition can almost always be satisfied  by choosing appropriate a $(\delta_1,\delta_2,\ldots,\delta_L)$ vector, except a few special cases where an all positive solution does not exist (recall we have assumed $\delta_\ell>0$, and thus only all positive solutions are valid). If the third condition is satisfied for certain source-channel-code triple, then it is satisfied for any more noisy channels.  It is seen that the critical conditions in the outer bound derivation essentially decouple the matching problem into several simpler ones, leading to the three largely independent conditions given in Theorem \ref{theorem:MAC}.

\subsection{Matched Channels in Special Case Scenarios}
\label{sec:MACexample}

In the multiple-access setting, the conditions for matching in Theorem \ref{theorem:MAC} are already rather simple, and there is no need to further investigate the properties of matched channels as in the broadcast case. Next we consider two special cases in the general problem setting which extend those considered in \cite{Gastpar:08} and \cite{Lapidoth:10MAC}, respectively.

\subsubsection{The Scalar CEO Problem}
\begin{figure}[tb]
\centering
\includegraphics[width=14cm]{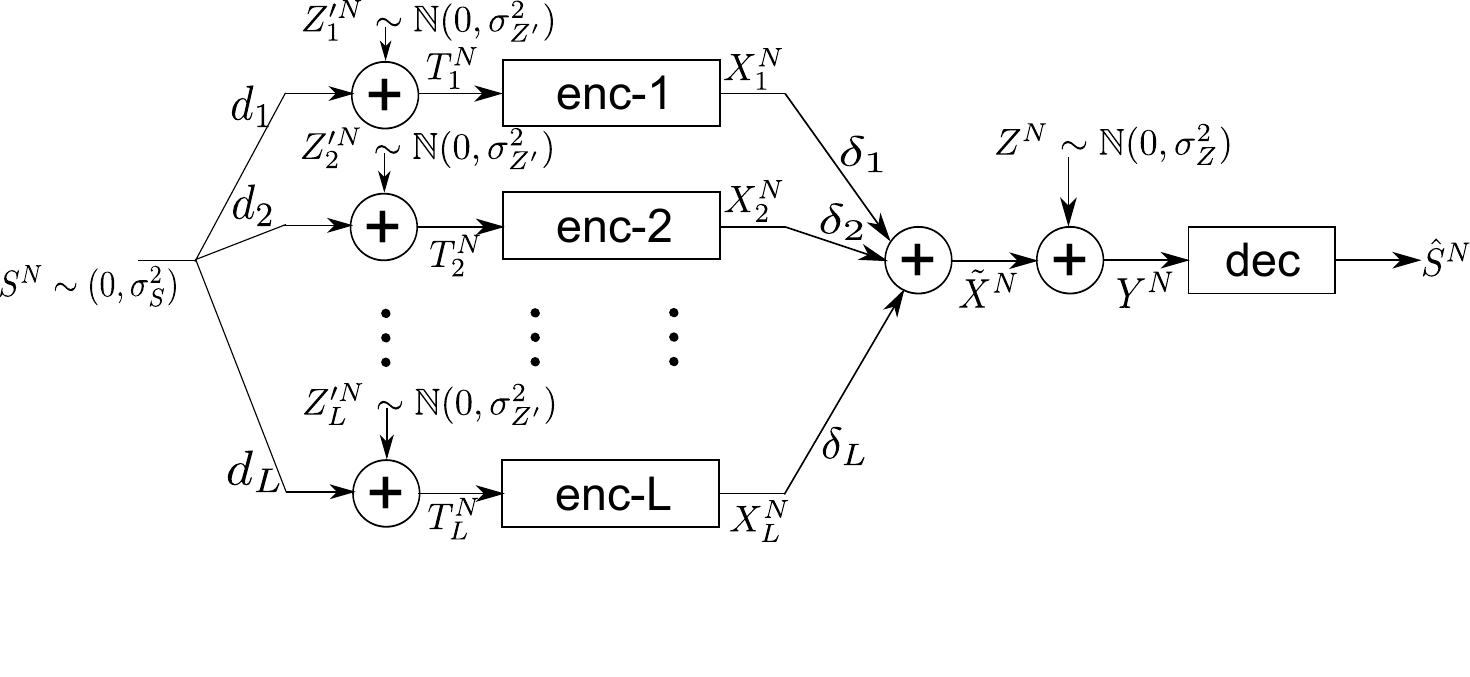}
\vspace{-0.5cm}
\caption{The scalar Gaussian CEO problem on a Gaussian multiple-access channel.  \label{fig:CEO}}
\end{figure}

Consider a zero-mean scalar Gaussian source $S[n]$ with covariance $\sigma^2_S$. There are a total of $L$ sensors, whose observations are
\begin{align}
T_\ell[n]=d_\ell S[n]+Z\rq{}_\ell[n], \quad \ell=1,2,\ldots,L,\quad n=1,2,\ldots,N,
\end{align}
where $d_\ell\geq 0$ (without loss of generality) and $Z\rq{}_\ell[n]$\rq{}s are the zero-mean independent additive noise with covariance $\sigma^2_{Z\rq{}}$. This special case is depicted in Fig. \ref{fig:CEO}.

It is clear that the first condition in Theorem \ref{theorem:MAC} is satisfied by $\eta_\ell=1$ for all $\ell=1,2,\ldots,L$.
The second condition for this case is equivalent to
\begin{align}
\left(\delta_1 \sqrt{\frac{P_1}{\psi_{1,1}}},\delta_2 \sqrt{\frac{P_2}{\psi_{2,2}}},\ldots,\delta_L \sqrt{\frac{P_L}{\psi_{L,L}}}\right)\Sigma_{T_{[1:L]}}\propto (d_1,d_2,\ldots,d_L),
\end{align}
where $\propto$ here means a component-wise proportional relation.
In other words, the uncoded scheme is optimal if
\begin{align}
\delta_\ell\sqrt{P_\ell(d^2_\ell\sigma^2_S+\sigma^2_{Z\rq{}})}+\sum_{\ell'\neq \ell}\delta_{\ell'}\sqrt{\frac{P_{\ell\rq{}}}{d^2_{\ell'}\sigma^2_S+\sigma^2_{Z\rq{}}}}d_\ell d_{\ell'}\sigma^2_S\propto d_\ell.
\end{align}
However, the LHS of the above condition can be simplified to
\begin{align}
\delta_\ell\sigma^2_{Z\rq{}}\sqrt{\frac{P_\ell}{d^2_\ell\sigma^2_S+\sigma^2_{Z'}}}+d_\ell\sum_{\ell'=1}^L\delta_{\ell'}d_{\ell'}\sqrt{\frac{P_{\ell'}}{d^2_{\ell'}\sigma^2_S+\sigma^2_{Z'}}}\sigma^2_S,
\end{align}
where the second term is proportional to $d_\ell$, and the first term is proportional to $d_\ell$ if and only if
\begin{align}
\frac{P_\ell\delta^2_\ell}{(d^2_\ell\sigma^2_S+\sigma^2_{Z'})d^2_\ell}=\mbox{const},\quad  \ell=1,2,\ldots,L.
\label{eqn:proportional}
\end{align}

It remains to check the third condition, however in this case $M=1$, and the second eigenvalue of the matrix $\Pi\Sigma_{\tilde{S}_{[1:M]}}\Pi$ can be viewed as zero, thus any noise power $\sigma^2_Z$ will allow a matching. Summarizing the above analysis, it is seen that for the scalar CEO problem on a Gaussian multiple-access channel, as long as the condition (\ref{eqn:proportional}) holds, the uncoded scheme is optimal. Conversely, for any noisy observation qualities, there always exists a matched channel by choosing the values of $\delta_\ell$ properly.

The condition  (\ref{eqn:proportional}) corresponds to a proportional quality requirement: the quality of the observations need to match the transmission powers and the transmission amplification factors. Gastpar \cite{Gastpar:08} showed that when all the sensors have the same observation quality, the same power and the same amplification factor, the uncoded scheme is optimal. Our result thus generalizes it to the proportional case.

\subsubsection{Correlated Gaussian Sources on a Gaussian Multiple-Access Channel}

Consider the case when $M=L$, and we shall assume that the first condition in Theorem \ref{theorem:MAC} can be satisfied. The second condition is also satisfied trivially since the matrix $\Sigma_{S_{[1:M]},T_{[1:L]}}$ is full rank in our problem setting. Thus only the last condition needs to be checked in this case. Equivalently, when $\lambda_2$ is strictly less than $1$, there always exists a noise power $\sigma^2_Z$ such that the channel is matched and thus the uncoded scheme is optimal.

\begin{figure}[tb]
\centering
\includegraphics[width=16cm]{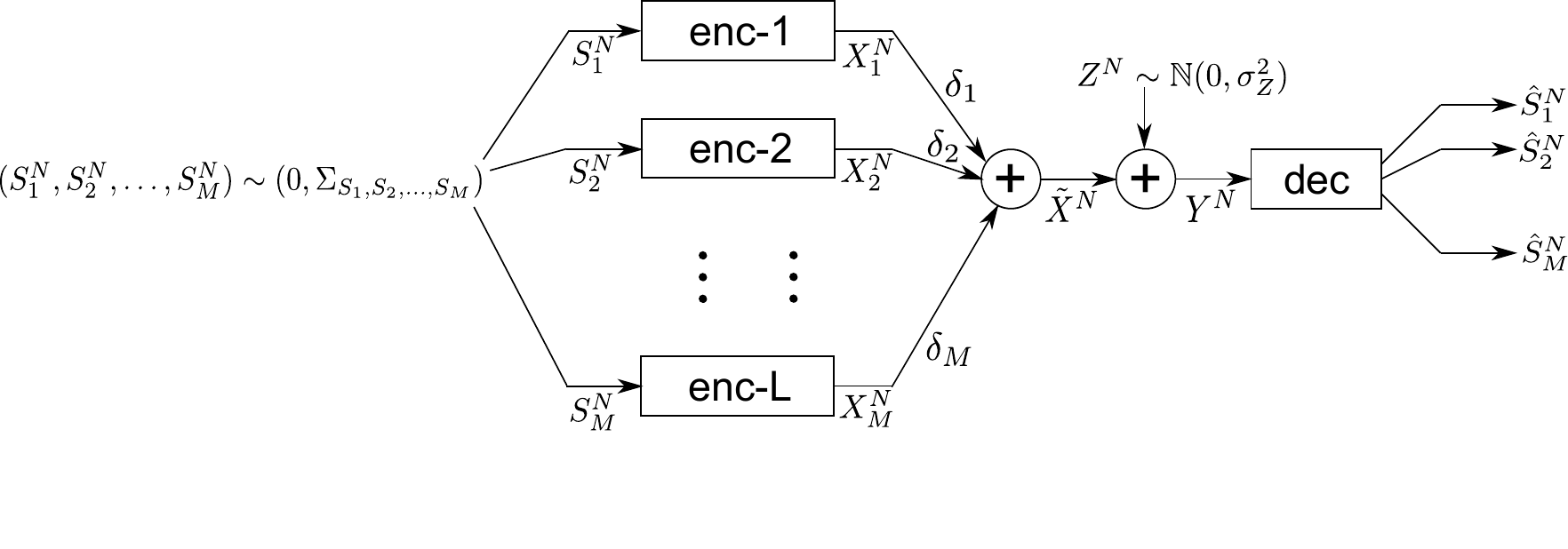}
\vspace{-0.5cm}
\caption{Sending correlated Gaussians on a Gaussian multiple-access channel.  \label{fig:LOCAL}}
\end{figure}

Lapidoth and Tinguely \cite{Lapidoth:10MAC} previously considered the special case when in addition to $M=L$, the observations are in fact noiseless and furthermore $T_m=S_m$, $m=1,2,\ldots,M$; see Fig. \ref{fig:LOCAL}. It was shown that for covariance matrix $\Sigma_{S_{[1:M]}}$ with strictly positive entries, there always exists a noise power $\sigma^2_Z$ such that the uncoded scheme is optimal. Our result generalizes theirs to the case that the observations can be noisy linear combinations, and the covariance matrix $\Sigma_{S_{[1:M]}}$ does not necessarily all have strictly positive entries.

\section{Conclusion}

We considered the problem of determining whether a given uncoded scheme is optimal for multiuser joint source channel coding. It was shown that for both broadcast and multiple-access in the Gaussian setting, matching occurs naturally under certain general conditions. Our approach differs from the more conventional approach in that instead of attempting to find explicit outer bound and inner bound then compare them, our focus is on the critical conditions that make the outer bound hold with equality. This approach has a decoupling effect which significantly simplifies the overall task. As future work, we plan to extend and generalize this approach to explore matching in other channel networks, and also for more general hybrid digital-analog schemes, for example, in the simple setting considered in  \cite{TianShamai:08,TianShamai:11}.

\begin{appendices}

\section{Proof of the Outer Bound in Theorem \ref{theorem:BC}}
\label{appendix:theorem1}
\begin{proof}

To upper-bound $E(\Sigma_{W_{[1:M]}})$, first recall the Markov string
\begin{align}
&Y^N_1\leftrightarrow Y^N_2 \leftrightarrow\ldots\leftrightarrow Y^N_M\leftrightarrow X^N \leftrightarrow (S^N_1,S^N_2,\ldots,S^N_M)\nonumber\\
&\qquad\leftrightarrow (U^N_1,U^N_2,\ldots,U^N_M)\leftrightarrow (U^N_1,U^N_2\leftrightarrow\ldots\leftrightarrow U^N_{M-1})\leftrightarrow\ldots\leftrightarrow U^N_1.\label{eqn:markovstring}
\end{align}
We start by writing the following:
\begin{align}
&\sum_{j=1}^m I(U^N_j;Y^N_j|U^N_1,U^N_2,\ldots,U^N_{j-1})\nonumber\\
&=\sum_{j=1}^m \left[I(U^N_1,U^N_2,\ldots,U^N_j;Y^n_j)-I(U^N_1,U^N_2,\ldots,U^N_{j-1};Y^N_j)\right]\nonumber\\
&=\sum_{j=1}^m \left[h(Y^N_j|U^N_1,U^N_2,\ldots,U^N_{j-1})-h(Y^N_j|U^N_1,U^N_2,\ldots,U^N_j)\right]\nonumber\\
&=\sum_{j=1}^m h(Y^N_j|U^N_1,U^N_2,\ldots,U^N_{j-1})-\sum_{j=1}^m h(Y^N_j|U^N_1,U^N_2,\ldots,U^N_j).
\end{align}
Since physical degradedness is equivalent to stochastic degradedness in the broadcast setting, {\em i.e.,} $Z_j$ can be assumed to be decomposable into two independent components as $Z_{j+1}+\Delta Z_j$, we can apply the entropy power inequality \cite{CoverThomas} for $j=1,2,\ldots,M-1$, 
\begin{align}
&\exp\left[\frac{2}{N}h(Y^N_j|U^N_1,U^N_2,\ldots,U^N_j)\right]\nonumber\\
&\geq \exp\left[\frac{2}{N}h(Y^N_{j+1}|U^N_1,U^N_2,\ldots,U^N_j)\right]+\exp\left[\log(2\pi e (\sigma^2_{Z_j}-\sigma^2_{Z_{j+1}}))\right]\nonumber\\
&=\exp\left[\frac{2}{n}h(Y^N_{j+1}|U^N_1,U^N_2,\ldots,U^N_j)\right]+2 \pi e(\sigma^2_{Z_j}-\sigma^2_{Z_{j+1}}). \label{eqn:applyentropypower}
\end{align}
For $j=M$, it is clear that
\begin{align}
&\exp\left[\frac{2}{N}h(Y^N_M|U^N_1,U^N_2,\ldots,U^N_M)\right]\geq \exp\left[\frac{2}{N}h(Y^N_M|S^N_1,S^N_2,\ldots,S^N_M)\right]= 2 \pi e\sigma^2_{Z_M},
\end{align}
with equality if and only if
\begin{align}
h(Y^N_M|U^N_1,U^N_2,\ldots,U^N_M)=h(Y^N_M|S^N_1,S^N_2,\ldots,S^N_M).\label{eqn:BCasifnonoisea}
\end{align}
It now follows that
\begin{align}
&E(\Sigma_{W_{[1:M]}})\nonumber\\
&=\sum_{m=1}^M (\sigma^2_{Z_{m}}-\sigma^2_{Z_{m+1}}) \exp\left[\frac{2}{N}\sum_{j=1}^m I(U^N_j;Y^N_j|U^N_1,U^N_2,\ldots,U^N_{j-1})\right]\nonumber\\
&\leq \sum_{m=1}^M (\sigma^2_{Z_{m}}-\sigma^2_{Z_{m+1}}) \frac{\exp\left[\frac{2}{N}\sum_{j=1}^m h(Y^N_j|U^N_1,U^N_2,\ldots,U^N_{j-1})\right]}{\prod_{j=1}^m\left[\exp\left(\frac{2}{N}h(Y^N_{j+1}|U^N_1,U^N_2,\ldots,U^N_j)\right)+2 \pi e (\sigma^2_{Z_j}-\sigma^2_{Z_{j+1}})\right]},
\end{align}
where for convenience we have defined $\exp\left[\frac{2}{N}h(Y^N_{M+1}|U^N_1,U^N_2,\ldots,U^N_M)\right]\triangleq 0$.

We upper-bound this summation by considering the summands in the reversed order, \textit{i.e.}, $m=M,M-1,\ldots,1$.
Starting with the summands when $m=M-1$ and $m=M$, we have 
\begin{align}
&(\sigma^2_{Z_{M-1}}-\sigma^2_{Z_{M}}) \frac{\exp\left[\frac{2}{N}\sum_{j=1}^{M-1} h(Y^N_j|U^N_1,U^N_2,\ldots,U^N_{j-1})\right]}{\prod_{j=1}^{M-1}\left[\exp\left(\frac{2}{N}h(Y^N_{j+1}|U^N_1,U^N_2,\ldots,U^N_j)\right)+2 \pi e(\sigma^2_{Z_j}-\sigma^2_{Z_{j+1}})\right]}\nonumber\\
&\qquad\qquad+\sigma^2_{Z_{M}} \frac{\exp\left[\frac{2}{N}\sum_{j=1}^{M} h(Y^N_j|U^N_1,U^N_2,\ldots,U^N_{j-1})\right]}{\prod_{j=1}^{M}\left[\exp\left(\frac{2}{N}h(Y^N_{j+1}|U^N_1,U^N_2,\ldots,U^N_j)\right)+2 \pi e(\sigma^2_{Z_j}-\sigma^2_{Z_{j+1}})\right]}\nonumber\\
&=\frac{\exp\left[\frac{2}{N}\sum_{j=1}^{M-1} h(Y^N_j|U^N_1,U^N_2,\ldots,U^N_{j-1})\right]}{\prod_{j=1}^{M-1}\left[\exp\left(\frac{2}{N}h(Y^N_{j+1}|U^N_1,U^N_2,\ldots,U^N_j)\right)+2 \pi e (\sigma^2_{Z_j}-\sigma^2_{Z_{j+1}})\right]}\nonumber\\
&\qquad\qquad\cdot\left[(\sigma^2_{Z_{M-1}}-\sigma^2_{Z_{M}}) +\sigma^2_{Z_{M}}\frac{\exp\left[\frac{2}{N}h(Y^N_M|U^N_{1},U^N_{2},\ldots,U^N_{M-1})\right]}{2 \pi e\sigma^2_{Z_{M}}}\right]\nonumber\\
&=\frac{1}{2 \pi e}\frac{\exp\left[\frac{2}{N}\sum_{j=1}^{M-1} h(Y^N_j|U^N_1,U^N_2,\ldots,U^N_{j-1})\right]}{\Pi_{j=1}^{M-2}\left[\exp\left(\frac{2}{N}h(Y^N_{j+1}|U^N_1,U^N_2,\ldots,U^N_j)\right)+2 \pi e(\sigma^2_{Z_{M-1}}-\sigma^2_{Z_{M}})\right]}.\label{eqn:firststep}
\end{align}
Continuing this line of reduction, we finally arrive at when $m=1$
\begin{align}
&E(\Sigma_{W_{[1:M]}})\nonumber\\
&\leq (\sigma^2_{Z_{1}}-\sigma^2_{Z_{2}})\frac{\exp\left[\frac{2}{N} h(Y^N_1)\right]}{\exp\left(\frac{2}{N}h(Y^N_{2}|U^N_1)\right)+2 \pi e(\sigma^2_{Z_{1}}-\sigma^2_{Z_{2}})}+\frac{1}{2 \pi e}\frac{\exp\left[\frac{2}{N}\sum_{j=1}^{2} h(Y^N_j|U^N_1,\ldots,U^N_{j-1})\right]}{\exp\left(\frac{2}{N}h(Y^N_{2}|U^N_1)\right)+2 \pi e(\sigma^2_{Z_{2}}-\sigma^2_{Z_{1}})}\nonumber\\
&=\frac{\exp\left[\frac{2}{N} h(Y^N_1)\right]}{\exp\left(\frac{2}{N}h(Y^N_{2}|U^N_1))\right)+2 \pi e(\sigma^2_{Z_{1}}-\sigma^2_{Z_{2}})}\left[(\sigma^2_{Z_{1}}-\sigma^2_{Z_{2}})+\frac{\exp \left[\frac{2}{N}h(Y^N_2|U^N_1)\right]}{2\pi e}\right]\nonumber\\
&=\frac{\exp\left[\frac{2}{N} h(Y^N_1)\right]}{2\pi e}\leq P+\sigma^2_{Z_1},\label{eqn:BCchannelA}
\end{align}
where the last inequality is by the concavity of the $\log(\cdot)$ function and the given power constraint. The chain of inequalities in (\ref{eqn:BCchannelA}) holds with equality   holds if and only if
\begin{align}
h(Y^N_1)=\frac{N}{2}\log 2\pi e( P+\sigma^2_{Z_1}),
\end{align}
as well as (\ref{eqn:BCasifnonoisea}) and the entropy power inequalities hold with equality.


%

We next lower bound $E(\Sigma_{W_{[1:M]}})$. By the rate-distortion theorem \cite{CoverThomas}
\begin{align}
I(U^N_1;Y^N_1)\geq \frac{N}{2}\log\frac{\sigma^2_{S_1}+\sigma^2_{W_1}}{D+\sigma^2_{W_1}},
\end{align}
with equality holds if and only if
\begin{align}
h(U^N_1|Y^N_1)=\frac{N}{2}\log [2\pi e(D_1+\sigma^2_{W_1} )].
\end{align}
Furthermore,
\begin{align}
&I(U^N_j;Y^N_j|U^N_1,U^N_2,\ldots,U^N_{j-1})\nonumber\\
&=h(U^N_j|U^N_1,U^N_2,\ldots,U^N_{j-1})-h(U^N_j|Y^N_j,U^N_1,U^N_2,\ldots,U^N_{j-1})\nonumber\\
&=\frac{N}{2}\log\frac{|2\pi e(\Sigma_{S_{[1:j]}}+\Sigma_{W_{[1:j]}})|}{|2\pi e(\Sigma_{S_{[1:j]}}+\Sigma_{W_{[1:j]}})|}-h(U^N_j|Y^N_j,U^N_1,U^N_2,\ldots,U^N_{j-1})\nonumber\\
&\geq \frac{N}{2}\log\frac{|2\pi e(\Sigma_{S_{[1:j]}}+\Sigma_{W_{[1:j]}})|}{|2\pi e(\Sigma_{S_{[1:j-1]}}+\Sigma_{W_{[1:j-1]}})|}-h(U^N_j|Y^N_j)\label{eqn:BCcondition2a}\\
&\geq\frac{N}{2}\log\frac{|2\pi e(\Sigma_{S_{[1:j]}}+\Sigma_{W_{[1:j]}})|}{|2\pi e(\Sigma_{S_{[1:j-1]}}+\Sigma_{W_{[1:j-1]}})|}-\frac{N}{2}\log [2\pi e(D_j+\sigma^2_{W_j} )]\label{eqn:BCcondition3a}\\
&=\frac{N}{2}\log\frac{|\Sigma_{S_{[1:j]}}+\Sigma_{W_{[1:j]}}|}{|(\Sigma_{S_{[1:j-1]}}+\Sigma_{W_{[1:j-1]}})|[D_m+\sigma^2_{W_j}]},\nonumber
\end{align}
where (\ref{eqn:BCcondition2a}) is because conditioning reduces entropy, and (\ref{eqn:BCcondition3a}) is because Gaussian distribution maximizes the differential entropy for random variables with the same variance \cite{CoverThomas}, together with the concavity of the $\log$ function. For (\ref{eqn:BCcondition2a}) to hold with equality, we must have
\begin{align}
h(U^N_j|Y^N_j,U^N_1,U^N_2,\ldots,U^N_{j-1})=h(U^N_j|Y^N_j),\quad j=2,3,\ldots,M,
\end{align}
and for (\ref{eqn:BCcondition3a}) to hold with equality it requires
\begin{align}
h(U^N_j|Y^N_j)=\frac{N}{2}\log [2\pi e(D_j+\sigma^2_{W_j} )],\quad j=2,3,\ldots,M.
\end{align}

It follows that
\begin{align}
\exp\left[\frac{2}{N}\sum_{j=1}^mI(U^N_j;Y^N_j|U^N_1,U^N_2,\ldots,U^N_{j-1})\right]\geq \frac{|\Sigma_{S_{[1:m]}}+\Sigma_{W_{[1:m]}}|}{\Pi^m_{j=1}(D_j+\sigma^2_{W_j})}.\label{eqn:BCsourceA}
\end{align}

Combining (\ref{eqn:BCchannelA}) and (\ref{eqn:BCsourceA}),  we reach an outer bound
\begin{align}
\sum_{m=1}^M  (\sigma^2_{Z_{m}}-\sigma^2_{Z_{m+1}})\frac{|\Sigma_{S_{[1:m]}}+\Sigma_{W_{[1:m]}}|}{\Pi^m_{j=1}(D_j+\sigma^2_{W_j})}\leq P+\sigma^2_{Z_1}.
\end{align}
\end{proof}

\section{Proof of Lemma 1}
\label{appendix:lemma1}
\begin{proof}

For simplicity, let us define $B^{(0)}_m\triangleq \sigma^2_{X|Y_m}$ for $m=1,2,\ldots,M$.
 It is clear that
\begin{align}
B^{(0)}_1\geq B^{(0)}_2\geq\ldots\geq B^{(0)}_M>0.
\end{align}

Recall $\alpha_m\neq 0$. In the $k$-th step of the Cholesky decomposition $k=0,1,\ldots,M-1$, we claim that $\alpha_{M-k}\beta_{M-k}\geq 0$ and $\sum_{i=1}^{M-k-1} \alpha_i\beta_i\geq0$. 
Moreover, we claim the matrix partially diagonalized, denoted as $\Sigma^{(k)}_{V_{[1:M]}}$, has entries in the following form:
\begin{itemize}
\item $\gamma^{(k)}_{i,j}=0$, $j>M-k$ and $i\neq j$; by symmetry, $\gamma_{i,j}=0$, $i> M-k$ and $i\neq j$;
\item $\gamma^{(k)}_{m,m}=\frac{\beta_m}{\alpha_m}B^{(k)}_m\sum_{j=1}^{m-1}\alpha_j\beta_j$, $m> M-k$;
\item $\gamma^{(k)}_{i,j}=-\beta_i\beta_j B^{(k)}_j$, $j\leq M-k$ and $i<j$; by symmetry $\gamma^{(k)}_{i,j}=-\beta_i\beta_j B^{(k)}_i$, $i\leq M-k$ and $i>j$;
\item $\gamma^{(k)}_{m,m}=\frac{\beta_m}{\alpha_m}\left[B^{(k)}_m\sum_{j=1}^{m-1}\alpha_j\beta_j+\sum_{j=m+1}^{M-k}\alpha_j\beta_jB^{(k)}_j\right]$, $m\leq M-k$;
\end{itemize}
where the terms $(B^{(k)}_1,B^{(k)}_2,\ldots,B^{(k)}_M)$ are determined recursively as
\begin{align}
B^{(k)}_m=B^{(k-1)}_m,\quad m> M-k,
\end{align}
and
\begin{align}
B^{(k+1)}_m=\left\{
\begin{array}{ll}
B^{(k)}_m+\frac{\alpha_{M-k}\beta_{M-k}}{\sum_{i=1}^{M-k-1}\alpha_i\beta_i}B^{(k)}_{M-k},&\beta_m\sum_{i=1}^{M-k-1}\alpha_i\beta_i\neq 0\\
B^{(k)}_m,& \mbox{otherwise}
\end{array}
\right.\quad m\leq M-k,
\end{align}
for which
\begin{align}
B^{(k)}_1\geq B^{(k)}_2\geq\ldots\geq B^{(k)}_M.
\end{align}

The readers can verify $\gamma^{(k+1)}_{i,j}$'s are precisely the expression when using Cholesky factorization on the matrix with entries $\gamma^{(k)}_{i,j}$'s.  
First consider the case $k=0$. Setting $m=M$ in (\ref{eqn:gammamm}) gives
\begin{align}
\gamma_{M,M}&=\alpha^{-1}_M\beta_M(\sum_{m=1}^{M-1}\alpha_m\beta_m)\frac{P\sigma^2_{Z_M}}{P+\sigma^2_{Z_M}}. 
\end{align}
Recall the assumption that $\alpha_M\neq 0$. The matrix $\Sigma^{(0)}_{V_{[1:M]}}$ being positive semi-definite implies that $\gamma_{M,M}\geq 0$, and since $\sum_{m=1}^M \alpha_m\beta_m=1$, it follows that $\alpha_M\beta_M(1-\alpha_M\beta_M)\geq 0$, and thus $\alpha_M\beta_M\in [0,1]$, from which we have  $\sum_{m=1}^{M-1}\alpha_m\beta_m\geq 0$. Thus the claim is true when $k=0$. Next suppose it is also true for $k=k^*$, and we wish to prove the claim for $k=k^*+1$.

It is clear that due to the positive semidefinite requirement for the degenerate case when
\begin{align}
\frac{\beta_{M-k^*}}{\alpha_{M-k^*}}\sum_{j=1}^{{M-k^*}-1}\alpha_j\beta_j=0,
\end{align}
we must have for $i< M-k^*$
\begin{align}
\gamma^{(k^*)}_{i, M-k^*}=B_{M-k^*}^{(k)}\beta_i\beta_{M-k^*}=0,
\end{align}
and this Cholesky step can essentially be skipped, and $(B^{(k)}_1,B^{(k)}_2,\ldots,B^{(k)}_M)$ does not need to be updated. It is easy to check the recursive formula $\gamma^{(k^*+1)}_{m,m}$ for $m\leq M-k^*-1$ is indeed valid for this case. 

If $\gamma^{(k^*)}_{M-k^*,M-k^*}\neq 0$, then due to the assumption in the induction we have
\begin{align}
\alpha_{M-k^*}\beta_{M-k^*}>0,\qquad \sum_{j=1}^{{M-k^*}-1}\alpha_j\beta_j>0.
\end{align}
First observe that due to the assumption in the induction, we have
\begin{align}
B^{(k^*+1)}_1\geq B^{(k^*+1)}_2\geq\ldots\geq B^{(k^*+1)}_M>0.
\end{align}
Using the Cholesky factorization, we have for any $j\leq M-k^*-1$ and $i<j$
\begin{align}
\gamma^{(k^*+1)}_{i,j}&=\gamma^{(k^*)}_{i,j}-\beta_i\beta_{M-k^*}B^{(k^*)}_{M-k^*}\frac{\alpha_{M-k^*}\beta_j}{\sum_{t=1}^{{M-k^*-1}}\alpha_t\beta_t}\nonumber\\
&=-\beta_i\beta_j B^{(k^*)}_j-\beta_i\beta_{M-k^*}B^{(k^*)}_{M-k^*}\frac{\alpha_{M-k^*}\beta_j}{\sum_{t=1}^{{M-k^*-1}}\alpha_t\beta_t}\nonumber\\
&=-\beta_i\beta_j\left[B^{(k^*)}_j+\frac{\alpha_{M-k^*}\beta_{M-k^*}}{\sum_{t=1}^{{M-k^*-1}}\alpha_t\beta_t}B^{(k^*)}_{M-k^*}\right]\nonumber\\
&=-\beta_i\beta_jB^{(k^*+1)}_j.
\end{align}
Similarly for $m\leq M-k^*-1$
\begin{align}
\gamma^{(k^*+1)}_{m,m}&=\gamma^{(k^*)}_{m,m}-\beta^2_m\beta_{M-k^*}B^{(k^*)}_{M-k^*}\frac{\alpha_{M-k^*}}{\sum_{t=1}^{{M-k^*-1}}\alpha_t\beta_t}\nonumber\\
&=\frac{\beta_m}{\alpha_m}\left[B^{(k^*)}_m\sum_{j=1}^{m-1}\alpha_j\beta_j+\sum_{j=m+1}^{M-k^*}\alpha_j\beta_jB^{(k^*)}_j\right]-\beta^2_m\beta_{M-k^*}B^{(k^*)}_{M-k^*}\frac{\alpha_{M-k^*}}{\sum_{j=1}^{{M-k^*-1}}\alpha_j\beta_j}\nonumber\\
&=\frac{\beta_m}{\alpha_m}B^{(k^*)}_m\sum_{j=1}^{m-1}\alpha_j\beta_j+\frac{\beta_m}{\alpha_m}\sum_{j=m+1}^{M-k^*-1}\alpha_j\beta_jB^{(k^*)}_j\nonumber\\
&\qquad+\alpha_{M-k^*}\beta_{M-k^*}\frac{\beta_m}{\alpha_m}\frac{\sum_{j=1}^{{M-k^*-1}}\alpha_j\beta_j-\alpha_m\beta_m}{\sum_{j=1}^{{M-k^*-1}}\alpha_j\beta_j}B^{(k^*)}_{M-k^*}\nonumber\\
&=\frac{\beta_m}{\alpha_m}\left[B^{(k^*)}_m+\frac{\alpha_{M-k^*}\beta_{M-k^*}}{\sum_{t=1}^{{M-k^*-1}}\alpha_t\beta_t}B^{(k^*)}_{M-k^*}\right]\sum_{j=1}^{m-1}\alpha_j\beta_j\nonumber\\
&\qquad+\frac{\beta_m}{\alpha_m}\sum_{j=m+1}^{M-k^*-1}\alpha_j\beta_j\left[B^{(k^*)}_j+\frac{\alpha_{M-k^*}\beta_{M-k^*}}{\sum_{t=1}^{{M-k^*-1}}\alpha_t\beta_t}B^{(k^*)}_{M-k^*}\right]\nonumber\\
&=\frac{\beta_m}{\alpha_m}\left[B^{(k^*+1)}_m\sum_{j=1}^{m-1}\alpha_j\beta_j+\sum_{j=m+1}^{M-k^*-1}\alpha_j\beta_jB^{(k^*+1)}_j\right].
\end{align}

Now suppose ${\alpha_{M-k^*-1}}{\beta_{M-k^*-1}}<0$, which implies that
\begin{align}
\sum_{j=1}^{{M-k^*}-2}\alpha_j\beta_j=\sum_{j=1}^{{M-k^*}-1}\alpha_j\beta_j-{\alpha_{M-k^*-1}}{\beta_{M-k^*-1}}>0.
\end{align}
This however contradicts with the positive semidefinite requirement that
\begin{align}
\frac{\beta_{M-k^*-1}}{\alpha_{M-k^*-1}}\sum_{j=1}^{{M-k^*}-2}\alpha_j\beta_j\geq 0.
\end{align}
Thus the supposition ${\alpha_{M-k^*-1}}{\beta_{M-k^*-1}}<0$ cannot be true.
If ${\alpha_{M-k^*-1}}{\beta_{M-k^*-1}}=0$, then from the assumption in the induction, we have $\sum_{j=1}^{{M-k^*}-2}\alpha_j\beta_j=\sum_{j=1}^{{M-k^*}-1}\alpha_j\beta_j\geq 0$ thus this case does not cause any problem.
If ${\alpha_{M-k^*-1}}{\beta_{M-k^*-1}}\geq0$, then it also follows that $\sum_{j=1}^{{M-k^*}-2}\alpha_j\beta_j\geq 0$. The lemma is proved.
\end{proof}

\section{Proof of Corollary \ref{coro1}}
\label{appendix:coro1}
\begin{proof}
It suffices to consider the case that  $\sigma^2_{Z^+_m}= \sigma^2_{Z_m}$, $m=1,2,\ldots,m^*-1,m^*+1,\ldots,M$, and $\sigma^2_{Z^+_{m^*}}= \sigma^2_{Z_{m^*}}+\sigma^2_{\Delta Z}$. Denote $\Delta P=\frac{P\sigma^2_{Z^+_{m^*}}}{P+\sigma^2_{Z^+_{m^*}}}-\frac{P\sigma^2_{Z_{m^*}}}{P+\sigma^2_{Z_{m^*}}}$, and matrices constructed for the two channels as $\Sigma_{V_{[1:M]}}$ and $\Sigma^*_{V_{[1:M]}}$, respectively.
It is clear that
\begin{align}
\Sigma^*_{V_{[1:M]}}-\Sigma_{V_{[1:M]}}=\begin{pmatrix}
    \frac{\beta_1}{\alpha_1}\alpha_{m^*}\beta_{m^*}\Delta P&0&\ldots&-\beta_1\beta_{m^*}\Delta P& 0&\ldots&0\\
    0&\frac{\beta_2}{\alpha_2}\alpha_{m^*}\beta_{m^*}\Delta P&\ldots&-\beta_2\beta_{m^*}\Delta P& 0&\ldots&0\\
    ...\\
   -\beta_1\beta_{m^*}\Delta P&   -\beta_2\beta_{m^*}\Delta P&\ldots&\frac{\beta_{m^*}}{\alpha_{m^*}}\Delta P\sum_{i=1}^{m^*-1}\alpha_i\beta_i&0&\ldots&0\\,
   0& 0&\ldots&0&0&\ldots&0\\
   \vdots& \vdots&\vdots&\vdots&\vdots&\vdots&\vdots\\
   0& 0&\ldots&0&0&\ldots&0\\
\end{pmatrix}.\label{eqn:incremental}
\end{align}
However, it is easily seen that this matrix is positive semidefinite since the first $m^*-1$ diagonal terms are non-negative, and we can remove all the other terms through symmetric elimination, {\em i.e.,} the Cholesky factorization step.
It follows that
\begin{align}
&\Sigma^*_{V_{[1:M]}}-\Sigma_{S_{[1:M]}}+P\bar{\beta}\bar{\beta}^t
\nonumber\\
&=[\Sigma^*_{V_{[1:M]}}-\Sigma_{V_{[1:M]}}]
+\left[\Sigma_{V_{[1:M]}}-\Sigma_{S_{[1:M]}}+P\bar{\beta}\bar{\beta}^t
\right]
\end{align}
is positive semidefinite since it is a summation of two  positive semidefinite matrices.
\end{proof}

\section{Proof of Corollary \ref{coro2}}

\begin{proof}
First note that the entries in matrix $\Pi$,
\begin{align}
\sum_{i=1}^M\rho_{j,i}\alpha_i=P\beta_j,\quad j=1,2,\ldots,M.
\end{align}

For the ``if'' direction, we choose a $\sigma^2_{Z}$ such that (\ref{eqn:Nlarge}) holds, which is always possible when $\sigma^2_{Z}$ is sufficiently large. This implies that for the channel $\sigma^2_{Z_1}=\sigma^2_{Z_2}=\ldots=\sigma^2_{Z_M}=\sigma^2_{Z}$, condition (\ref{eqn:matchingspecial}) holds, and thus it is a matched channel.

For the ``only if'' direction, it follows from Corollary 1 that matching must hold for the degraded channel with noise power $\sigma^2_{Z_1'}=\sigma^2_{Z_2'}=\ldots=\sigma^2_{Z_M'}\triangleq \sigma^2_{Z}=\sigma^2_{Z_1}$. The requirement (\ref{eqn:semidefinite}) implies
\begin{align}
\frac{P\sigma^2_{Z}}{P+\sigma^2_{Z}}\mbox{diag}\left(\frac{\beta_1}{\alpha_1},\frac{\beta_2}{\alpha_2},\ldots,\frac{\beta_M}{\alpha_M}\right)+\left(P-\frac{P\sigma^2_{Z}}{P+\sigma^2_{Z}}\right)\bar{\beta}\bar{\beta}^t
\succeq \Sigma_{S_1,S_2,\ldots,S_M}. \label{eqn:matchingspecial}
\end{align}
This, together with Lemma 1, implies that $\alpha_i\beta_i>0$, since otherwise the left hand side is a rank deficient. Multiplying both sides of (\ref{eqn:matchingspecial}) from left and from right by $\Pi$ gives
\begin{align}
\frac{\sigma^2_{Z}}{P+\sigma^2_{Z}} I_M+\frac{P}{P+\sigma^2_{Z}}\begin{pmatrix}
    \sqrt{\alpha_1\beta_1}\\
    \sqrt{\alpha_2\beta_2}\\
    ...\\
    \sqrt{\alpha_M\beta_M}
\end{pmatrix}\begin{pmatrix}
   \sqrt{\alpha_1\beta_1}& \sqrt{\alpha_2\beta_2}& ...&   \sqrt{\alpha_M\beta_M}
\end{pmatrix}\succeq \Pi\Sigma_{S_{[1:M]}}\Pi.
\end{align}
Notice that $\bar{v}^t_1=(\sqrt{\alpha_1\beta_1}, \sqrt{\alpha_2\beta_2}, ...,   \sqrt{\alpha_M\beta_M})$ is in fact an eigenvector corresponding to the eigenvalue $1$ for the matrix $\Pi\Sigma_{S_{[1:M]}}\Pi$, easily verified using (\ref{eqn:alphatobeta}). We can write the eigen-decomposition of the matrix $\Pi\Sigma_{S_{[1:M]}}\Pi$ as
\begin{align}
\Pi\Sigma_{S_1,S_2,\ldots,S_M}\Pi=\bar{v}_1\bar{v}^t_1+\sum_{i=2}^M\lambda_i\bar{v}_i\bar{v}^t_i,
\end{align}
where $\lambda_2,\lambda_3,\ldots,\lambda_M$ are the other eigenvalues of $\Pi\Sigma_{S_1,S_2,\ldots,S_M}\Pi$, and $\bar{v}_2,\bar{v}_3,\ldots,\bar{v}_M$ are the corresponding eigenvectors. It follows that
\begin{align}
\frac{\sigma^2_{Z}}{P+\sigma^2_{Z}} I_M\succeq \frac{\sigma^2_{Z}}{P+\sigma^2_{Z}}\bar{v}_1\bar{v}^t_1+\sum_{i=2}^M\lambda_i\bar{v}_i\bar{v}^t_i, \label{eqn:Nlarge}
\end{align}
which implies $\lambda_i\leq \frac{\sigma^2_{Z}}{P+\sigma^2_{Z}}$, $i=2,3,\ldots,M$.

\end{proof}

\section{Proof of Corollary \ref{coro3}}
\label{appendix:coro3}
\begin{proof}
Consider the entries of matrix $\mbox{diag}(\bar{\alpha})\Sigma_{W_1,W_2,\ldots,W_M}\mbox{diag}(\bar{\alpha})$, denoted as $\phi_{i,j}$, which is given as (by symmetry only the upper-triangle entries need to be specified)
\begin{align}
\phi_{j,m}=\alpha_j\beta_j\alpha_m\beta_m\frac{P^2}{P+\sigma^2_{Z_m}}-\alpha_j\alpha_m\rho_{j,m},\quad j<m,
\end{align}
and
\begin{align}
\phi_{m,m}=\alpha_m\beta_m\sum_{j=1}^{m-1}\alpha_j\beta_j\frac{P\sigma^2_{Z_m}}{P+\sigma^2_{Z_m}}+\alpha_m\beta_m\sum_{j=m+1}^M\alpha_j\beta_j\frac{P\sigma^2_{Z_j}}{P+\sigma^2_{Z_j}}-\alpha^2_m\rho_{m,m}+\alpha^2_m\beta^2_mP.
\end{align}
A necessary and sufficient condition for matching is that the matrix $\mbox{diag}(\bar{\alpha})\Sigma_{W_{[1:M]}}\mbox{diag}(\bar{\alpha})$ is positive semidefinite, since this implies the existence of the required $(W_1,W_2,\ldots,W_M)$ random vector, or equivalently the required random vector $(V_1,V_2,\ldots,V_M)$ as in the proof of  Theorem \ref{theorem:BC}. 

Observe that
\begin{align}
\sum_{j=1}^M\phi_{j,m}=\alpha_m\beta_mP\sum_{i=1}^{M}\alpha_j\beta_j - \alpha_m\beta_mP=0.\label{eqn:sumtozero}
\end{align}
If all the off-diagonal entries of $\mbox{diag}(\bar{\alpha})\Sigma_{W_{[1:M]}}\mbox{diag}(\bar{\alpha})$ are non-positive, then the matrix is diagonally dominant, and the diagonal entries are all positive by (\ref{eqn:sumtozero}), which implies that it is a positive semidefinite matrix \cite{matrixanalysis}. Thus as long as
\begin{align}
\alpha_j\beta_j\alpha_m\beta_m\frac{P^2}{P+\sigma^2_{Z_m}}\leq\alpha_j\alpha_m\rho_{j,m},\quad j<m \label{eqn:allcondition}
\end{align}
and
\begin{align}
\alpha_j\beta_j\alpha_m\beta_m\frac{P^2}{P+\sigma^2_{Z_j}}\leq\alpha_j\alpha_m\rho_{j,m},\quad j>m\label{eqn:allcondition2}
\end{align}
the positive semidefinite condition is satisfied. Note here that since $\mbox{diag}(\bar{\alpha})\Sigma_{S_{[1:M]}}\mbox{diag}(\bar{\alpha})$ has positive entries, $\alpha_i\beta_i>0$, and both sides of the above conditions are positive, which makes it possible for them to hold by choosing $\sigma^2_{Z_j}$'s properly. 
It is thus sufficient to have
\begin{align}
\sigma^2_{Z_m}\geq \max_{j< m} \frac{\beta_j}{\rho_{j,m}}\beta_mP^2-P,\quad m=2,3\ldots, M.
\end{align}
Together with Corollary 1, this implies the statement given in the corollary is indeed true.
\end{proof}

\section{Proof of (\ref{eqn:source2})}
\label{appendix:source2}
\begin{proof}
Since the matrix $\Sigma_{S_1,S_2,S_3}$ is positive definite, we have $-1<\rho_1<1$ and $-1<\rho_2<1$.
Since $\alpha_m=1$, we must also have $\beta_m>0$, $0<\lambda_2<1$ and $0<\lambda_3<1$ for matching to occur.
The first condition gives that
\begin{align}
\rho_1+\rho_2+1>0\quad\mbox{and}\quad 2\rho_1+1>0,
\end{align}
but the latter two require a few more steps. Notice that the condition $0<\lambda_3<1$ implies that
\begin{align}
\rho_1+2\rho_2>0.
\end{align}
If $\rho_2>2\rho^2_1-1$, then $0<\lambda_2<1$ implies
\begin{align}
2\rho_1\rho_2> -4\rho^2_1-3\rho_1.
\end{align}
If $\rho_1>0$, this yields a condition already implied by $\rho_1+2\rho_2>0$; on the other hand, $\rho_1\leq 0$ is an impossible case.
It can be verified that $\rho_2<2\rho^2_1-1$ is also an impossible case. Thus we must have $\rho_2>2\rho^2_1-1$ and $\rho_1>0$ simultaneously, from which we obtained the set of conditions given in (\ref{eqn:source2}).
\end{proof}

\end{appendices}

\section*{Acknowledgment}

The authors are extremely grateful for the insightful and detailed comments provided by the reviewers, which helped to improve the presentation of the paper.

\bibliographystyle{IEEEbib}

\end{document}